\newif\ifirmk\irmkfalse
\newcommand{\wasisdas}{{\sc Integrating Polchinski's equation with binary trees}}

\documentclass[11pt,reqno]{amsart} 
\usepackage[papersize={210mm,297mm},left=30mm,right=30mm,top=20mm,bottom=20mm]{geometry}
\newcommand{\randbem}[1]{\marginpar{\tiny\raggedright\wichtigst{\bf #1}}}
\usepackage{fancyhdr}
\usepackage{amssymb}
\usepackage{amsthm}
\usepackage{mathrsfs}
\usepackage{graphicx}

\usepackage{xcolor}
\definecolor{slategray}{RGB}{112,138,144}
\definecolor{hintergrundblau}{RGB}{0 ,43 ,54}
\definecolor{wichtigrot}{RGB}{220,50,47}
\definecolor{textgrau}{RGB}{131,148,150}
\definecolor{ultramarinblau}{RGB}{32,33,79}
\definecolor{kobaltblau}{RGB}{35,45,83}
\definecolor{navy}{RGB}{0 ,0,128}
\definecolor{leuchtgruen}{RGB}{0,181,26}
\definecolor{leuchtrot}{RGB}{255,77,6} 

\definecolor{firebrick}{RGB}{178,34,34}
\definecolor{darkorange}{RGB}{255,140,0}
\definecolor{darkgreen}{RGB}{0,100,0}
\definecolor{seagreen}{RGB}{46,139,87}
\definecolor{lightseagreen}{RGB}{32,178,170}
\definecolor{forestgreen}{RGB}{34,139,34}
\definecolor{midnightblue}{RGB}{25,25,112}
\definecolor{navyblue}{RGB}{0,0,128}
\definecolor{cornflowerblue}{RGB}{100,149,237}
\definecolor{mediumblue}{RGB}{0,0,205}
\definecolor{dimgray}{RGB}{105,105,105}
\definecolor{slategray}{RGB}{112,138,144}
\definecolor{Korrektur}{RGB}{0,0,205}

\usepackage[colorlinks=true]{hyperref}
\hypersetup{linkcolor=slategray,urlcolor=slategray,citecolor=slategray}

\newcommand{\E}{{\rm e}}

\newcommand{\dd}{{\rm d}}

\newcommand{\beq}{\begin{equation}}
\newcommand{\eeq}{\end{equation}}
\newcommand{\neeq}{\nonumber\end{equation}}
\newcommand{\bea}{\begin{eqnarray}}
\newcommand{\eea}{\end{eqnarray}}
\newcommand{\beast}{\begin{eqnarray*}}
\newcommand{\eeast}{\end{eqnarray*}}
\newcommand{\beal}{\begin{align}}
\newcommand{\eeal}{\end{align}}

\newcommand{\bbbone}{{\mathchoice {\rm 1\mskip-4mu l} {\rm 1\mskip-4mu l}    {\rm 1\mskip-4.5mu l} {\rm 1\mskip-5mu l}}}

\newcommand{\bN}{{\mathbb N}}

\newcommand{\de}{\delta}

\newcommand{\bGa}{{{\rm I}\kern-.16em \Gamma}}
\newcommand{\cV}{{\mathcal V}}

\newcommand{\abs}[1]{{\left\vert #1 \right\vert}}
\newcommand{\norm}[1]{{\left\Vert #1 \right\Vert}}

\newcommand{\sfrac}[2]{{\textstyle \frac{#1}{#2}}}

\newcommand{\wichtigst}[1]{{\color{red}#1}}

\usepackage{stmaryrd}   
\usepackage{enumitem}   
\usepackage{tikz}       
\usepackage{subcaption}
\usepackage{graphbox}

\usetikzlibrary{shapes.arrows}

\newcommand{\Z}{\mathbb{Z}}
\newcommand{\N}{\mathbb{N}}
\newcommand{\R}{\mathbb{R}}
\newcommand{\C}{\mathbb{C}}
\newcommand{\DD}{\mathcal{D}}

\newcommand{\set}[1]{\left\{#1\right\}}

\newcommand{\sgn}{\operatorname{sgn}}

\newcommand{\GL}{\operatorname{GL}}
\newcommand{\card}[1]{\##1}
\newcommand{\ordprod}[1]{\overset{\rightarrow}{\prod_{#1}}}
\newcommand{\discup}{\mathbin{\dot\cup}}

\newcommand{\vfork}{V_{\mathrm{fork}}}
\newcommand{\vleaf}{V_{\mathrm{leaf}}}
\newcommand{\leavesof}[2]{\mathcal{L}_{#1}\left(#2\right)}
\newcommand{\innerprod}[2]{\left(#1,#2\right)}

\newcommand{\G}{\mathcal{G}}
\newcommand{\Gr}{\mathfrak{G}}
\newcommand{\B}{\mathfrak{B}}
\newcommand{\T}{\mathfrak{T}}
\renewcommand\L{\mathfrak{L}}
\newcommand{\X}{\mathcal{X}}

\newcommand{\V}{\mathcal{V}}

\newcommand{\pdv}[2]{\frac{\partial #1}{\partial #2}}
\newcommand{\fdv}[2]{\frac{\delta #1}{\delta #2}}

\newcommand{\Diff}{\mathfrak{Diff}}

\theoremstyle{plain}
\newtheorem{theorem}{Theorem}
\newtheorem{lemma}[theorem]{Lemma}

\newtheorem{corollary}[theorem]{Corollary}

\theoremstyle{definition}
\newtheorem{definition}[theorem]{Definition}

\theoremstyle{remark}
\newtheorem{remark}[theorem]{Remark}
\newtheorem{notation}[theorem]{Notation}

\renewcommand{\subseteq}{\subset}
\renewcommand{\card}[1]{\abs{#1}}
\newcommand{\Mcov}{\mathbb{M}}

\begin{document}

\title[]{Integrating Polchinski's equation \\ by convergent binary tree expansions}
\author{Paul Obernolte, Manfred Salmhofer} 
\address{ 
Institut f\" ur Theoretische Physik, Universit\" at Heidelberg,
Philosophenweg~19, 69120 Heidelberg, Germany}

\email{P.Obernolte@thphys.uni-heidelberg.de,salmhofer@uni-heidelberg.de}

\date{\today}

\begin{abstract}
\noindent
We give a solution to Polchinski's equation for the Wilsonian effective action in terms of an expansion in binary trees, and prove that this expansion converges in fermionic field theories, provided the fermionic covariance has finite determinant and decay constants. A novel element of the proof are detailed combinatorial estimates for the number of leaf trees associated to binary trees. The method can be used on standard models of fermionic quantum field theory and quantum statistical mechanics.  
\end{abstract}

\maketitle

%

\section{Introduction}
\label{sec:intro}

Introduced in statistical mechanics in the study of critical phenomena, the Kadanoff-Wilson-Wegner renormalization group (RG) \cite{Kadanoff,Wilson1,Wilson2,Wegner} has become a standard in the theory of quantum fields and quantum many-body systems, both in theoretical physics and mathematics. It sets up a flow in a space of action functionals or Hamiltonians, which is parametrized by a scale, such as a length or an energy. The basic idea is that averaging over local fluctuations at a certain length scale defines a scale-dependent {\em effective action}, which contains all information about correlations at larger length scales, and that by varying the scale parameter one can understand how fundamental or `microscopic' models and emergent or `macroscopic' models are related. A specific model to be studied then corresponds to an initial condition, starting from which the flow tends to effective actions that characterize the long-distance correlations of the model. If, as is the case in many examples, the flow tends to an attractor, typically a fixed point, the corresponding fixed-point effective action determines the emergent model and its properties. 

Mathematically, the local averaging corresponds to a smoothing operation which also gives a natural way of regularizing quantum field theoretical models, and field-theoretic renormalization can be understood as a way of adjusting the initial condition of a Wilsonian RG flow in order to get a finite limit at the end of the flow. A particularly nice way of doing this in perturbative QFT was given by Polchinski \cite{polchinskiRenormalizationEffectiveLagrangians1984}. Following his seminal work, the method was simplified and extended to general renormalization conditions \cite{KKS}, and it has subsequently been applied to a large variety of models of QFT, both related to high-energy physics \cite{KeKoQED,KeKoZi,KoMue,HoStM} and to condensed-matter physics \cite{crg}. 

The averaging operation has been implemented in technically diverse ways, from the block spin transformation originally introduced by Kadanoff \cite{Kadanoff} to momentum and phase space variants. Mathematical proofs have mostly been done using discrete scale sequences, where the scale parameter (such as the above-mentioned length) is increased in steps, e.g.\ in a geometric progression. Formal theory and approximate calculations have very often been done using a continuous variation of the scale parameter. The mathematical proofs with discrete RG sequences include pioneering results like the construction of infrared $\phi^4$ theory in four dimensions \cite{GK1,FMRS1}, the first constructions of the two-dimensional Gross-Neveu model \cite{GK2,FMRS2}, and the construction of a two-dimensional Fermi liquid at zero temperature \cite{FKT}. 

In particular, Polchinski's equation has been studied in many cases and in different representations. Given a continuous flow, one can of course also derive from it a discrete ``flow'', i.e.\ a sequence of actions, simply as the values of the flowing effective action at a sequence of scales. In their seminal work, Brydges and Kennedy analyzed Polchinski's equation by PDE methods and derived an explicit solution as an expansion labelled by trees. They used it to give a new proof of convergence of Mayer expansions \cite{BryKe}. This was subsequently applied to fermions by Brydges and Wright \cite{BryWri}. Abdesselam and Rivasseau gave important generalizations and combinatorial perspectives in \cite{AR}. Disertori and Rivasseau used the integrated flow equation to construct the Gross-Neveu model \cite{DRGN} and the rotationally symmetric two-dimensional many-fermion system \cite{DR1,DR2}. 

In this paper we integrate Polchinski's equation by a binary tree expansion, which arises naturally because the nonlinearity in the differential equation is quadratic. The form it takes is a bit different from the Brydges-Kennedy formula, and indeed one can think of it as a continuous version of the Gallavotti-Nicol\`o expansion, which was introduced and used in the context of discrete scale decompositions and applied to proofs of perturbative renormalizability of four-dimensional scalar field theories \cite{GN}, and further developed and applied to quantum electrodynamics in \cite{FHRW}. 

For the case of fermionic field theories with regular covariances, we prove that our binary-tree expansion converges. Regularity of the covariance means that both its determinant bound and its decay constant are finite, and bounded independently of the volume. Tree expansions have been used to analyze fermionic theories in many works, in particular in the context of many-fermion systems \cite{BGPS,Pedra,BGM}. For a prototypical development of the technique see, e.g.\ \cite{SW}; bounds on determinant constants especially suitable for time-ordered covariances were given in \cite{PSUV}.  One novel aspect of our proof are combinatorial estimates for ``leaf trees'' that are induced by the binary trees arising in the iterative solution of the Polchinski equation, and by the action of Laplacians in field space in that context.

\section{Setup and main results}
\label{sec:setup-and-main-results}

\subsection{Grassmann algebra setup}

We will consider a fermionic QFT on a finite lattice. Take the lattice to be the torus $\Gamma = \epsilon \Z^d / L \Z^d$, where $d$ is the dimension, $\epsilon > 0$, and $L \gg 1$ such that $L / (2 \epsilon) \in \N$. After doing all computations, we want to take the continuum limit ($\epsilon \to 0$), as well as the thermodynamic limit ($L \to \infty$). In addition, let there be a finite (but non-empty) set $I$ of internal indices; these can be spin indices or other indices. Our field variables $\theta_X$ will then be indexed by $X \in \mathcal{X} = \Gamma \times I$ with $N=\#\mathcal{X}$. We consider the vector space $V = \R^\mathcal{X}$ and its Grassmann algebra
\begin{equation}
\bigwedge (V) = \bigoplus_{i=0}^N {\bigwedge}^i(V),
\end{equation}
which is generated by
\begin{align*}
\theta_{X_1} \wedge \dots \wedge \theta_{X_k}
\end{align*}
with $k \in \N$ and $X_1, \dots, X_k \in \mathcal{X}$. For the sake of simplicity, we will write $\theta_{X_1} \dots \theta_{X_k}$ for $\theta_{X_1} \wedge \dots \wedge \theta_{X_k}$ in the following. We have the commutation relation $\theta_X \theta_Y = -\theta_Y \theta_X$. 
We shall also consider Grassmann algebras that include source terms, i.e.\ algebras of the form $\bigwedge (V) \wedge \bigwedge (V)$, and denote the corresponding Grassmann generators by $\psi$.

We define a norm on $\bigwedge (V)$. 

\begin{definition}[$h$-norm on $\bigwedge (V)$]
\label{def:h-norm}
Let $h \in \R_{>0}$. A general element $\V$ of $\bigwedge (V)$ can be written as
\begin{equation}
\V(\Psi) = \sum_{m \geq 0} \int \dd^m{\underline{X}}\ v_m(\underline{X}) \psi^m(\underline{X}),
\end{equation}
where $\int \dd^m{\underline{X}} = \epsilon^{md} \sum_{\underline{X} \in \mathcal{X}^m}$, $\psi^m(\underline{X}) = \psi(X_1) \dots \psi(X_m)$, and every $v_m$ is an antisymmetric $m$-tensor. Then, $\norm{\V}_h$ shall be defined as
\begin{equation}
\norm{\V}_h = \sum_{m \geq 1} \abs{v_m} h^m,
\end{equation}
where $\abs{v_m}$ is the standard norm\footnote{For $k \in \N$, we use the notation $[k] = \set{1, \dots, k}$ and $[k]-1 = \set{0, \dots, k-1}$.}
\begin{equation}
\abs{v_m} = \max_{i \in [m]} \sup_{X_i \in \mathcal{X}} \int \prod_{j \in [m] \setminus \set{i}} \dd{X_j}\; \abs{v_m(X_1, \dots, X_m)}.
\end{equation}
\end{definition}

\begin{remark}
Note that the ``$h$-norm'' $\norm{.}_h$ is not actually a norm on $\bigwedge (V)$, but only a semi-norm, because $\norm{c}_h = 0$ for all $c \in \C = {\bigwedge}^0(V)$. However, $\norm{.}_h$ is a norm on the subspace
\begin{eqnarray}
\bigoplus_{i=1}^N {\bigwedge}^i(V) \subsetneq \bigwedge (V).
\end{eqnarray}
\end{remark}

We assume the action of our QFT to be of the form
\begin{align}
\label{eq:qft-action}
S(\Theta) = \frac{1}{2} (\Theta, Q \Theta) + \mathcal{V}(\Theta),
\end{align}
where $Q \in \GL_N(\R)$ is a skew-symmetric matrix, $\mathcal{V}$ is a potential, and the bilinear form $(\cdot,\cdot)$ is defined by:
\begin{align}
(\Theta, \Psi) = \int_\mathcal{X}{\dd{X} \; \theta_X \psi_X} = \epsilon^d \sum_{X \in \mathcal{X}}{\theta_X \psi_X}.
\end{align}
Now, we can define the generating function for the connected Green functions:
\begin{align}
W(\Psi) = -\log \int \DD{\Theta}\; e^{-S(\Theta) + (\Theta,\Psi)}.
\label{eq:generating-function}
\end{align}
We call $e^{-W(\Psi)}$ the \textit{partition function} of our theory. The functional integral in the previous expression is defined as
\begin{align}
\int \DD{\Theta} = \prod_{X \in \mathcal{X}} \int \dd{\theta_X},
\end{align}
where $\int \dd{\theta_X}$ is a Grassmann integral.\footnote{Note that the order of the integrals matters for the overall sign; hence, we need an order on $\mathcal{X}$. As introduced above, $\mathcal{X} = \Gamma \times I$, where $I$ is a set of internal indices. Usually, $\# I$ is even; this means that we only need an ordering of $I$ and then, the ordering of $\Gamma$ does not matter.} In our case, rather than studying $W(\Psi)$ it is more natural to study the Wilsonian effective potential defined by
\begin{equation}
\label{eq:wilsonian-effective-potential}
\G(\Psi) = - \log \int \dd{\mu_C}(\Theta)\; e^{-\mathcal{V}(\Theta+\Psi)} = - \log\ ({\mu_C} * e^{-\mathcal{V}})(\Psi),
\end{equation}
where $C = Q^{-1}$ and ${\mu_C}$ denotes the Grassmann Gaussian measure with covariance $C$,  $\dd \mu_C (\Theta) = \mathcal{N} e^{-\frac12 (\Theta, Q \Theta)} \prod_X \dd \Theta_X$ (with $\mathcal{N}$ the normalizing factor). If we want to highlight the dependence of $\mathcal{G}(\Psi)$ on $C$ and $\mathcal{V}$, we might also occasionally write $\mathcal{G}(\Psi,C,\mathcal{V})$ for it. By completing the square in the exponent, one easily proves that studying $W(\Psi)$ and $\G(\Psi)$ is equivalent.

We introduce the notation $\nabla = \frac{\de}{\de \Psi} = \epsilon^{-d} \pdv{}{\Psi}$, with $\pdv{}{\Psi}$ the standard Grassmann derivative,  and define the field space Laplacian as follows. For a skew-symmetric $D$,  
\beq
\Delta_{D} = \innerprod{\nabla}{D \nabla} 
=
\int_\mathcal{X} \dd{X} \int_\mathcal{X} \dd{Y} \; D (X,Y) \nabla_X \nabla_Y \; .
\eeq
With this, we then have \cite{msbook, Hesselberg}
\beq
\E^{-\G(\Psi) }
= 
\E^{\frac12 \Delta_C}  \; e^{-\mathcal{V}(\Psi)} \; .
\eeq
Because the Grassmann derivatives themselves form a Grassmann algebra of finite dimension, the exponential of the Laplacian truncates to polynomial, so this expression is well defined. Moreover, it makes obvious that the limit $C \to 0$ of the convolution exists and becomes the identity operator. Thus $\G(\Psi,0,\cV) = \cV(\Psi)$. 

We note that all the Laplacians we use in this paper are constant coefficient operators so they all commute: $\Delta_D \Delta_C = \Delta_C \Delta_D$ for all $C,D$. 

\subsection{Polchinski's equation}

In the framework of renormalization theory, it is natural to integrate degrees of freedom associated to different length or energy scales successively. Mathematically, this is done by decomposing a covariance (which, in absence of regularizations, typically has an integral kernel with singularities) into many regular pieces. The decomposition can be discrete or continuous; we are interested in the latter and write
\begin{equation}
C(s) = \int_0^s \dd{r}\; \dot{C}(r),
\end{equation}
where $s \in \R$ plays the role of the \textit{flow parameter}. Assume that $C(s)$ is a smooth function of $s$ and that every $\dot{C}(r)$ is skew-symmetric. Note that a dot will always represent a derivative with respect to the flow parameter. We will frequently use the following notation for $s,t \in \R$:
\begin{equation}
C_{t,s} = \int_s^t \dd{r}\; \dot{C}(r).
\end{equation}
Now, we can compute the Wilsonian effective potential based on the covariance matrix $C(s)$. This way, the Wilsonian effective potential also gains a dependence on the flow parameter $s$. We write $\mathcal{G}(s,\Psi) = \mathcal{G}(\Psi, C(s), \mathcal{V})$ for short; sometimes, we even write $\mathcal{G}(s) = \mathcal{G}(s,\Psi)$.
By taking the derivative of $\G(s)$ with respect to $s$, one gets Polchinski's equation \cite{polchinskiRenormalizationEffectiveLagrangians1984}
\begin{equation}
\label{eq:polchinski-equation}
\dot{\mathcal{G}}(s,\Psi) = \frac{1}{2}\Delta_{\dot{C}(s)} \mathcal{G}(s,\Psi) - \frac{1}{2} \left( \nabla \mathcal{G}(s,\Psi), \dot{C}(s) \nabla \mathcal{G}(s,\Psi) \right).
\end{equation}
Polchinski's equation is a first order ordinary differential equation (ODE) in the flow parameter $s$ on the finite-dimensional space $\bigwedge (V)$. Given an initial condition, it has a unique solution because Gaussian convolution is well-defined on $\bigwedge (V)$ whenever $C(s)$ is well-defined.

We define $C:\ [0,1] \to M_N(\R)$ in such a way that $C(0) = 0$. Then $\G(0, \Psi) = \V(\Psi)$ where $\V$ is the initial interaction of our model. Hence, our goal will be to find the Wilsonian effective potential as a solution of Polchinski's equation with this initial condition.

To derive our expansion, we start from the following integral version of Polchinski's initial value problem (see also \cite[Appendix~A]{kroschinskyMajorantMethodFermionic2025}). 

\begin{lemma}
\label{lem:integrate-polchinski}
Polchinski's Equation~\eqref{eq:polchinski-equation} with initial condition $\G(0,\Psi) = \V(\Psi)$ is equivalent to the integrated Polchinski equation
\begin{equation}
\label{eq:integrated-polchinski-equation}
\G(t, \Psi) = \G_0(t,\Psi) - \int_0^t \dd{s}\; e^{\frac{1}{2}\Delta_{C_{t,s}}} \cdot \frac{1}{2} \innerprod{\nabla \G(s,\Psi)}{\dot{C}(s) \nabla \G(s,\Psi)},
\end{equation}
where
\begin{equation}
\G_0(t,\Psi) = e^{\frac{1}{2} \Delta_{C_{t,0}}} \V(\Psi).
\end{equation}
\end{lemma}

\begin{proof}
Because $\frac{\dd}{\dd s} \Delta_{C_{t,s}} = -\Delta_{\dot C(s)}$ and $\Delta_{C_{t,s}}$ commutes with $\Delta_{\dot C(s)}$, 
\begin{equation}
\begin{split}
\G(t) - \E^{\frac12 \Delta_{C_{t,0}}} \G(0)
&=
\int_0^t \dd s \; \frac{\dd}{\dd s} 
\left[
\E^{\frac12 \Delta_{C_{t,s}}} \G(s)
\right]
\\
&=
\int_0^t \dd s \; \E^{\frac12 \Delta_{C_{t,s}}}
\left[
-\sfrac12 \Delta_{\dot C(s)} \G(s) + \dot \G(s)
\right]
\\
&=
- \frac12 
\int_0^t \dd s \; \E^{\frac12 \Delta_{C_{t,s}}}\;
(\nabla \G(s), \dot C (s) \nabla \G(s)) \; .
\end{split}
\end{equation}
In the last step we have used (\ref{eq:polchinski-equation}).
\end{proof}

\subsection{Finding a formal solution in binary trees}

We now want to write down a formal solution to the (integrated) Polchinski equation using binary trees. In order to do so, we first need to define binary trees and some notation related to them. We refer to Appendix~\ref{app:graph-theory} for some basic notions and notations of graph theory.

\begin{definition}[(Planar) binary tree]
\label{def:binary-tree}
A \textit{(planar) binary tree} is a connected graph $T=(V,E)$ that fulfils the following properties:
\begin{enumerate}[label=(\roman*)]
\item $V \subseteq \set{(n,k) : n \in \N_0,\ k\in [2^n]-1}$ and $* = (0,0) \in V$,
\item $E \subseteq \set{\set{(n,k),(n+1,k')} : k \equiv k' \pmod{2^n}}$,
\item $d_T(*) \in \set{0,2}$ and for every $v \in V \setminus \set{*}$, $d_T(v) \in \set{1, 3}$.
\end{enumerate}
We call $* \in V$ the \textit{root} of $T$. $d_T (*) = 0$ only if $V(V)=\{*\}$, so that $T$ has no lines.
A sample binary tree including the vertex labels can be seen in Figure~\ref{fig:example-binary-tree-with-vertex-labels}.
We further define:
\begin{enumerate}[label=(\roman*)]
\item The \textit{leaves} of $T$ are the elements of
\begin{equation}
\vleaf(T) = \set{w \in V(T) : d_T(w) = 1}.
\end{equation}
\item The \textit{forks} of $T$ are the elements of
\begin{equation}
\vfork(T) = V(T) \setminus \vleaf(T).
\end{equation}
\item Write $l(T) = \card{\vleaf(T)}$ for the number of leaves of $T$.
\item The leaf mapping maps every vertex to the set of leaves under it:\footnote{By $\mathcal{P}(\vleaf(T))$, we denote the power set of $\vleaf(T)$.}
\begin{equation}
\begin{split}
\mathcal{L}_T:\ &V(T) \to \mathcal{P}(\vleaf(T)), \\
&(n,k) \mapsto \set{(n',k') \in \vleaf(T) : k \equiv k' \pmod{2^n}}.
\end{split}
\end{equation}
\item The parent mapping maps every vertex to its parent vertex:
\begin{equation}
\begin{split}
p_T:\ &V(T) \to V(T) \sqcup \set{\dagger}, \\
&(n,k) \mapsto \begin{cases}
\dagger & \text{if $n=0$} \\
(n-1, k \pmod{2^{n-1}}) & \text{if $n \geq 1$}
\end{cases}.
\end{split}
\end{equation}
Here, $\dagger$ is just a newly introduced symbol.
\item The children mappings map every fork to its left or right child in the tree; define it for $i \in \set{1,2}$:
\begin{equation}
\begin{split}
c_T^{(i)}:\ &\vfork(T) \to V(T), \\
&(n,k) \mapsto \begin{cases}
(n+1,k) & \text{if $i=1$} \\
(n+1,k+2^n) & \text{if $i=2$}
\end{cases}.
\end{split}
\end{equation}
\item Define a total order on $V(T)$ by
\begin{equation}
(n,k) < (n',k') \quad :\Leftrightarrow \quad n < n' \text{ or } (n=n' \text{ and } k < k').
\end{equation}
\end{enumerate}
Some of these preceding notations are showcased in Figure~\ref{fig:example-binary-tree-with-notation}.
Lastly, for $l \in \N$, use the notation
\begin{equation}
\mathfrak{B} = \set{T : \text{$T$ is a binary tree}}
\text{ and }
\mathfrak{B}_l = \set{T \in \mathfrak{B} : l(T) = l}.
\end{equation}
\end{definition}

\begin{figure}[h!]
\begin{subfigure}[t]{0.7\textwidth}
\begin{center}
\tiny
\begin{tikzpicture}
\node[anchor=south] at (3.5,3) {$(0,0) = *$};
\node[circle, fill=black, inner sep=0, minimum size=3pt] at (3.5,3) {};

\draw (3.5,3) -- (1.5,2) node[anchor=east] {$(1,0)$};
\node[circle, fill=black, inner sep=0, minimum size=3pt] at (1.5,2) {};
\draw (3.5,3) -- (5.5,2) node[anchor=west] {$(1,1)$};
\node[circle, fill=black, inner sep=0, minimum size=3pt] at (5.5,2) {};

\draw (1.5,2) -- (0.5,1) node[anchor=east] {$(2,0)$};
\node[circle, fill=black, inner sep=0, minimum size=3pt] at (0.5,1) {};
\draw (1.5,2) -- (2.5,1) node[anchor=west] {$(2,2)$};
\node[circle, fill=black, inner sep=0, minimum size=3pt] at (2.5,1) {};
\draw (5.5,2) -- (4.5,1) node[anchor=east] {$(2,1)$};
\node[circle, fill=black, inner sep=0, minimum size=3pt] at (4.5,1) {};
\draw (5.5,2) -- (6.5,1) node[anchor=west] {$(2,3)$};
\node[circle, fill=black, inner sep=0, minimum size=3pt] at (6.5,1) {};

\draw (0.5,1) -- (0,0) node[anchor=east] {$(3,0)$};
\node[circle, fill=black, inner sep=0, minimum size=3pt] at (0,0) {};
\draw[dotted] (0,0) -- (-0.25,-0.5);
\draw[dotted] (0,0) -- (0.25,-0.5);

\draw (0.5,1) -- (1,0) node[anchor=east] {$(3,4)$};
\node[circle, fill=black, inner sep=0, minimum size=3pt] at (1,0) {};
\draw[dotted] (1,0) -- (0.75,-0.5);
\draw[dotted] (1,0) -- (1.25,-0.5);

\draw (2.5,1) -- (2,0) node[anchor=east] {$(3,2)$};
\node[circle, fill=black, inner sep=0, minimum size=3pt] at (2,0) {};
\draw[dotted] (2,0) -- (1.75,-0.5);
\draw[dotted] (2,0) -- (2.25,-0.5);

\draw (2.5,1) -- (3,0) node[anchor=east] {$(3,6)$};
\node[circle, fill=black, inner sep=0, minimum size=3pt] at (3,0) {};
\draw[dotted] (3,0) -- (2.75,-0.5);
\draw[dotted] (3,0) -- (3.25,-0.5);

\draw (4.5,1) -- (4,0) node[anchor=east] {$(3,1)$};
\node[circle, fill=black, inner sep=0, minimum size=3pt] at (4,0) {};
\draw[dotted] (4,0) -- (3.75,-0.5);
\draw[dotted] (4,0) -- (4.25,-0.5);

\draw (4.5,1) -- (5,0) node[anchor=east] {$(3,5)$};
\node[circle, fill=black, inner sep=0, minimum size=3pt] at (5,0) {};
\draw[dotted] (5,0) -- (4.75,-0.5);
\draw[dotted] (5,0) -- (5.25,-0.5);

\draw (6.5,1) -- (6,0) node[anchor=east] {$(3,3)$};
\node[circle, fill=black, inner sep=0, minimum size=3pt] at (6,0) {};
\draw[dotted] (6,0) -- (5.75,-0.5);
\draw[dotted] (6,0) -- (6.25,-0.5);

\draw (6.5,1) -- (7,0) node[anchor=east] {$(3,7)$};
\node[circle, fill=black, inner sep=0, minimum size=3pt] at (7,0) {};
\draw[dotted] (7,0) -- (6.75,-0.5);
\draw[dotted] (7,0) -- (7.25,-0.5);
\end{tikzpicture}
\normalsize
\end{center}
\caption{Example binary tree with vertex labels.}
\label{fig:example-binary-tree-with-vertex-labels}
\end{subfigure}%
\begin{subfigure}[t]{0.3\textwidth}
\begin{center}
\tiny
\begin{tikzpicture}
\path[rounded corners=5, fill=blue, opacity=0.2] (2.5,-0.25) -- (1.5,-0.25) -- (1.25,0) -- (0.75,1) -- (1,1.25) -- (1.25,1) -- (1.5,0.25) -- (2.5,0.25) -- (2.75,0) -- cycle;

\node[anchor=south] at (2,4) {$*$};
\node[circle, fill=black, inner sep=0, minimum size=3pt] at (2,4) {};

\draw (2,4) -- (1,3) node[anchor=east] {$p_T(v)$};
\node[circle, fill=black, inner sep=0, minimum size=3pt] at (1,3) {};
\draw (2,4) -- (3,3);
\node[circle, fill=black, inner sep=0, minimum size=3pt] at (3,3) {};

\draw (1,3) -- (0.5,2);
\node[circle, fill=black, inner sep=0, minimum size=3pt] at (0.5,2) {};
\draw (1,3) -- (1.5,2) node[anchor=west] {$v \in V(T)$};
\node[circle, fill=black, inner sep=0, minimum size=3pt] at (1.5,2) {};

\draw (1.5,2) -- (1,1) node[anchor=east] {$c_T^{(1)}(v)$};
\node[circle, fill=black, inner sep=0, minimum size=3pt] at (1,1) {};
\draw (1.5,2) -- (2,1) node[anchor=west] {$c_T^{(2)}(v)$};
\node[circle, fill=black, inner sep=0, minimum size=3pt] at (2,1) {};

\draw (2,1) -- (1.5,0);
\node[circle, fill=black, inner sep=0, minimum size=3pt] at (1.5,0) {};
\draw (2,1) -- (2.5,0);
\node[circle, fill=black, inner sep=0, minimum size=3pt] at (2.5,0) {};

\node[anchor=west, blue] at (2.75,0) {$\leavesof{T}{v}$};
\end{tikzpicture}
\normalsize
\end{center}
\caption{Example binary tree $T$ showcasing some of the notation introduced in Definition~\ref{def:binary-tree} (i) to (vii).}
\label{fig:example-binary-tree-with-notation}
\end{subfigure}

\caption{Explanatory drawings for Definition~\ref{def:binary-tree}.}
\end{figure}
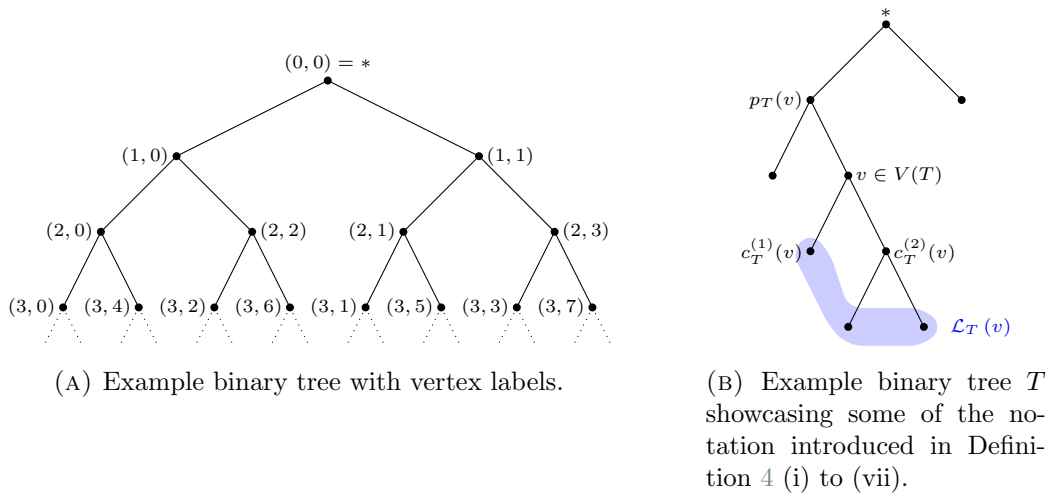

This now allows us to write down a solution of Polchinski's equation as a formal power series in the initial interaction $\G_0$. The expansion coefficients are labelled by binary trees. We then go on by proving its convergence (Theorem \ref{thm:convergence-of-formal-solution}).
 
\begin{theorem}
\label{thm:formal-solution}
The integrated Polchinski Equation~\eqref{eq:integrated-polchinski-equation} is solved by the formal series $\G(t,\Psi) = \sum_{l=1}^\infty \G_l(t,\Psi)$, where $\G_l(t,\Psi)$, the $l^{\rm th}$ order term in $\G_0$, is given by a sum over binary trees
\begin{equation}
\G_l(t,\Psi) = \sum_{T \in \mathfrak{B}_l} \G_T(t, \Psi)
\end{equation}
and for every $T \in \mathfrak{B}_l$
\begin{equation}
\begin{split}
\G_T(t, \Psi)
&= \Bigg\llbracket  \ordprod{v \in \vfork(T)} \int_0^{s_{p(v)}} \dd{s_v}\;  \Diff^{(T,v)}_{s_{p(v)},s_{\vphantom{(}v}} \; \prod_{w \in \vleaf(T)} \G_0\Big(s_{p(w)}, \Psi^{(w)}\Big)  \Bigg\rrbracket.
\end{split}
\end{equation}
with the differential operator
\beq
\Diff^{(T,v)}_{s_{p(v)},s_{\vphantom{(}v}}
=
-\sfrac{1}{2} \; 
e^{\frac{1}{2} \Delta_{C_{s_{p(v)},s_{v\phantom{)}}}}^{(\leavesof{T}{v})}} \; 
\innerprod{\nabla^{(\mathcal{L}_1(v))}}{\dot{C}(s_v) \nabla^{(\mathcal{L}_2(v))}} 
\eeq
Here, we have also introduced the following notation:
\begin{itemize}
\item We have introduced multiple copies of the field variables, namely one copy $\Psi^{(v)}$ for every leaf vertex $v \in \vleaf(T)$. For $A \subseteq \vleaf(T)$, we use the notation
\begin{align}
\Delta_C^{(A)} &= \sum_{i,j \in A} \Delta_C^{(i,j)} = \sum_{i,j \in A} \innerprod{\fdv{\;}{\Psi^{(i)}}}{C \fdv{\;}{\Psi^{(j)}}}, \\
\nabla^{(A)} &= \sum_{i \in A} \nabla^{(i)} = \sum_{i \in A} \fdv{\;}{\Psi^{(i)}}.
\end{align}
Lastly, we use special parentheses $\llbracket . \rrbracket$ to express that all copies of the field variables within these brackets shall be set the original field variables. In other words, we define
\begin{equation}
\llbracket . \rrbracket = [ . ]_{\Psi^{(v)}=\Psi\ \forall v \in \vleaf(T)}.
\end{equation}
\item The ordered product $\ordprod{}_{v \in \vfork(T)}$ uses the total order on $V(T)$ given in Definition~\ref{def:binary-tree}.
\item We write $p(v)$ instead of $p_T(v)$ and $c_i(v)$ (with $i \in \set{1,2}$) instead of $c_T^{(i)}(v)$ because the dependence on the tree $T$ is clear, similarly $\mathcal{L}_i(v) = \leavesof{T}{c_i(v)}$ for $i=1,2$.
\item We have set $s_{p(*)} = s_\dagger = t$.
\end{itemize}
\end{theorem} 

\begin{proof}
By substituting the formal power series into the right-hand side of Equation~\eqref{eq:integrated-polchinski-equation}, we obtain:
\begin{equation}
\begin{split}
&\G_0(t,\Psi) - \int_0^t \dd{s}\; e^{\frac{1}{2}\Delta_{C_{t,s}}} \cdot \frac{1}{2} \innerprod{\nabla \G(s,\Psi)}{\dot{C}(s) \nabla \G(s,\Psi)} \\
=& \G_0(t,\Psi) + \sum_{l_1,l_2=1}^\infty \sum_{\substack{T_1 \in \B_{l_1} \\ T_2 \in \B_{l_2}}} \int_0^t \dd{s}\; e^{\frac{1}{2} \Delta_{C_{t,s}}} \left(-\frac{1}{2}\right) \innerprod{\nabla \G_{T_2}(s,\Psi)}{\dot{C}(s) \nabla \G_{T_2}(s,\Psi)} \\
=& \G_0(t,\Psi) + \sum_{l_1,l_2=1}^\infty \sum_{\substack{T_1 \in \B_{l_1} \\ T_2 \in \B_{l_2}}} \G_{T}(t,\Psi),
\end{split}
\end{equation}
where $T \in \B_{l_1+l_2}$ is the binary tree obtained by combining $T_1$ and $T_2$ and connecting their respective roots to a new root. Here we have used that $\nabla \llbracket \prod_{i \in A} F_i (\Psi^{(i)}) \rrbracket  = \llbracket \nabla^{(A)} \prod_{i \in A} F_i (\Psi^{(i)})\rrbracket$. Further, we have:
\begin{equation}
\begin{split}
&\G_0(t,\Psi) + \sum_{l_1,l_2=1}^\infty \sum_{\substack{T_1 \in \B_{l_1} \\ T_2 \in \B_{l_2}}} \G_{T}(t,\Psi) \\
=& \G_{(\set{*},\emptyset)}(t,\Psi) + \sum_{l=2}^\infty \sum_{T \in \B_l} \G_T(t,\Psi) \\
=& \sum_{l=1}^\infty \sum_{T \in \B_l} \G_T(t,\Psi)
= \G(t,\Psi),
\end{split}
\end{equation}
where $(\set{*},\emptyset)$ denotes the trivial binary tree with exactly one leaf. This proves the statement.
\end{proof}

\subsection{Convergence of the tree expansion}

Here we show the convergence of the formal expansion stated in Theorem \ref{thm:formal-solution} in terms of the determinant bound $\gamma_C$ and the decay constant $\omega_C$  of the covariance matrix $C$. 

\begin{theorem}
\label{thm:convergence-of-formal-solution}
Let $h > 0$  and  $\G(t,\Psi) = \sum_{l=1}^\infty \G_l(t,\Psi)$ be the (formal) solution of Equation~\eqref{eq:integrated-polchinski-equation} from Theorem~\ref{thm:formal-solution}. Let $\dot{C}(t) = C$ for all $t \in [0,1]$, where $C \in M_N(\R)$ is a (constant) skew-symmetric matrix. Let $C$ have a Gram constant\footnote{See Definition~\ref{def:gram-constant}.} $\gamma_C > 0$ and let there be a $\omega_C > 0$ such that
\begin{equation}
|C| = \sup_{X \in \X} \int \dd{Y}\; |C(X,Y)| \leq \omega_C \gamma_C^2.
\end{equation}
Finally, let $\G (0) =\V$. Then:
\begin{equation}
\label{eq:bound-on-wilsonian-effective-potential-order-l}
\norm{\G_l(t=1)}_h \leq 2 \; (4e \; \omega_C)^{l-1} \; \norm{\V}_{h'}^l,
\end{equation}
where $h'=h + \left(\sqrt{2e} + 1\right) \gamma_C$ (and $e$ is Euler's constant). 

Moreover, if $4e \; \omega_C \; \norm{\V}_{h'} < 1$ then
\begin{equation}
\label{eq:bound-on-wilsonian-effective-potential}
\norm{\G(t=1)}_h
\leq \frac{2\norm{\V}_{h'}}{1 - 4e \; \omega_C \; \norm{\V}_{h'}}.
\end{equation}
In particular, the formal power series $\G(t,\Psi)$ is convergent in this case. For $P \geq 1$, we have in this case:
\begin{equation}
\label{eq:bound-on-wilsonian-effective-potential-remainder}
\norm{\G(t=1) - \sum_{l=1}^{P-1} \G_l(t=1)}_h
\leq \frac{2}{4e \; \omega_C} \; \frac{\big(4e \; \omega_C \; \norm{\V}_{h'}\big)^{P}}{1 - 4e \; \omega_C \; \norm{\V}_{h'}}.
\end{equation}
\end{theorem}

If the determinant bound $\gamma_C$ and the decay constant $\omega_C$  of the covariance matrix $C$ can be chosen independent of the lattice spacing $\epsilon$ and the system size $L$, then the tree expansion for the Wilsonian effective potential converges in a disk that is uniform in $\epsilon$ and $L$, hence the continuum limit and the thermodynamic limit can be taken. 

Theorem \ref{thm:convergence-of-formal-solution} is the central result of this paper. It will be proven in Section~\ref{sec:proof-of-convergence}. In the proof we will need some results on the combinatorics on binary trees, which we prove in the next section.

\section{Combinatorics of graphs}
\label{sec:combinatorics-of-graphs}

\subsection{Short review}

Graph expansions are a common tool in quantum field theory. Simple notions of graph theory are reviewed in Appendix~\ref{app:graph-theory}. To bound graph expansions, one has to know about the combinatorics of the specific type of graphs used. The simplest expansions of the Wilsonian (effective) potential that one can obtain are expansions in general (connected) graphs. Unfortunately, these have quite a bad growth behaviour in the number of vertices:

\begin{lemma}[Number of (connected) graphs on $n$ vertices]
\label{lem:number-of-graphs}
Let $n \in \N$.
\begin{enumerate}[label=(\roman*)]
\item The number of graphs on $n$ vertices is
\begin{equation}
\card{\Gr([n])} = 2^{\binom{n}{2}} = 2^{\frac{n(n-1)}{2}}.
\end{equation}
\item For $n \to \infty$, the number of connected graphs on $n$ vertices is 
\begin{equation}
\card{\Gr_c([n])} \sim 2^{\binom{n}{2}} = 2^{\frac{n(n-1)}{2}}.
\end{equation}
\end{enumerate}
\end{lemma}

\begin{proof}
Part (i) is clear. To prove part (ii), we use from \cite[Section~3.11]{generatingfunctionology} that $d_n$, the number of labelled connected graphs on $n$ vertices, satisfies
\begin{equation*}
n \; 2^{{n \choose 2}}
=
\sum_{k \ge 1} {n \choose k} k\; d_k \; 2^{{n - k \choose 2}}.
\end{equation*}
Define $\delta_k = d_k \; 2^{-{k\choose 2}}$. Then by (i), $\delta_k \le 1$ for all $k \in \bN$ and 
\begin{equation*}
1
=
\delta_n + (n-1) \delta_{n-1} 2^{1-n} + \sum_{k=1}^{n-2} {n \choose k} \frac{k}{n} \delta_k\; 2^{- k (n - \frac{k-1}{2})} \; .
\end{equation*}
The last sum is bounded by
\begin{equation*}
\sum_{k=1}^{n-2} {n \choose k} 2^{-k (n - \frac{k-1}{2})}
\le
\sum_{k\ge 1} {n \choose k} 2^{-k n/2 }
=
(1+ 2^{-n/2})^n -1
\le
\E^{n 2^{-n/2}} -1.
\end{equation*}
Thus
\begin{equation*}
\delta_n \ge
1 - (n-1) 2^{1-n} - (\E^{n 2^{-n/2}} -1) \; ,
\end{equation*}
so for $n \ge 10$, $\delta_n \ge 1 - n 2^{-n/2}$. 
\end{proof}

For this reason, expansions in general (connected) graphs usually behave badly in the sense that estimates for individual graphs are not sufficient to prove convergence of the series. Expansions in tree graphs behave a lot better because the number of trees on $n$ vertices exhibits a much smaller growth rate:

\begin{lemma}[Cayley's formula {\cite[Chapter~2]{lintCombinatorics}}]
\label{lem:cayleys-formula}
Let $n \in \N$ and $d_1, \dots, d_n \geq 1$ such that $\sum_{i=1}^n d_i = 2(n-1)$.
\begin{enumerate}[label=(\roman*)]
\item We have:
\begin{equation}
\card{\mathfrak{T}_n} = n^{n-2}.
\end{equation}
\item Define
\begin{equation}
\mathfrak{T}_n(d_1, \dots, d_n)
= \set{T \in \mathfrak{T}_n : d_T(i) = d_i\ \forall i \in [n] }.
\end{equation}
Then:
\begin{equation}
\card{\mathfrak{T}_n(d_1, \dots, d_n)} = \frac{(n-2)!}{(d_1 - 1)! \dots (d_n - 1)!}.
\end{equation}
\end{enumerate}
\end{lemma}

We have $n^{n-2} \leq n! e^n$; the factorial is usually cancelled by an inverse factor $1/n!$  in according tree expansions.

The expansion presented in this paper is in binary trees. Their number can be obtained by Catalan family considerations:

\begin{lemma}[Number of binary trees {\cite[Example~14.9]{lintCombinatorics}}]
\label{lem:num-of-binary-trees}
For $l \in \N$, we have $\card{\mathfrak{B}_l} = \frac{1}{l} \binom{2(l-1)}{l-1}$.
\end{lemma}

Hence, $\card{\mathfrak{B}_l} < (2e)^{l-1}$; the growth rate of binary trees is non-factorial, hence significantly better than even that of normal trees. On the other hand, it turns out that our expansion in binary trees gives rise to a sum over yet another structure: Leaf trees of binary trees. These are defined and analysed in the next subsection.

\subsection{Combinatorics of binary trees}

We begin by defining so-called leaf trees of binary trees.

\begin{definition}[(Induced) leaf tree of a binary tree]
\label{def:induced-leaf-tree}
Let $T$ be a binary tree. An \textit{(induced) leaf tree} of $T$ is a graph $L=(V,E)$ with $V = \vleaf(T)$, that fulfils the following property:
\begin{equation}
\begin{split}
&\forall v \in \vfork(T)\ \exists ! \set{w_1,w_2} \in E \\
&\text{ such that } w_1 \in \mathcal{L}_T(c_T^{(1)}(v)) \text{ and } w_2 \in \mathcal{L}_T(c_T^{(2)}(v)).
\end{split}
\end{equation}
For an induced leaf tree $L$ of $T$, define the mappings
\begin{equation}
c_L^{(i)}:\ \vfork(T) \to \vleaf(T)
\end{equation}
for $i \in \set{1,2}$ such that for any $v \in \vfork(T)$, we have $c_L^{(i)}(v) \in \mathcal{L}_T(c_T^{(i)}(v))$ (for $i \in \set{1,2}$) and $\set{c_L^{(1)}(v), c_L^{(2)}(v)} \in E(L)$.
We use the following notation:
\begin{equation}
\mathfrak{L}(T) = \set{L : \text{$L$ is an induced leaf tree of $T$}}.
\end{equation}
\end{definition}

\begin{remark}
Note that every induced leaf tree of a binary tree is indeed a tree, although we did not explicitly require it in the definition (one can show this by induction).
\end{remark}

In Figure~\ref{fig:leaf-trees}, we give an example for a binary tree $T$ and all leaf trees that it induces.

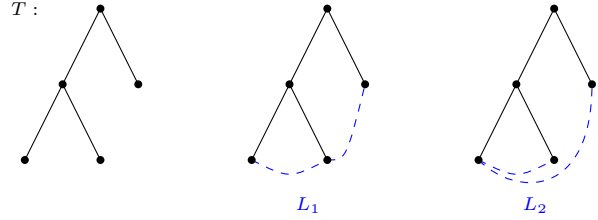
\begin{figure}[h!]
\tiny
\begin{tikzpicture}
\node at (0,2) {$T:$};

\draw (1,2) -- (0.5,1);
\draw (1,2) -- (1.5,1);

\draw (0.5,1) -- (0,0);
\draw (0.5,1) -- (1,0);

\node[circle, fill=black, inner sep=0, minimum size=3pt] at (1,2) {};
\node[circle, fill=black, inner sep=0, minimum size=3pt] at (0.5,1) {};
\node[circle, fill=black, inner sep=0, minimum size=3pt] at (1.5,1) {};
\node[circle, fill=black, inner sep=0, minimum size=3pt] at (0,0) {};
\node[circle, fill=black, inner sep=0, minimum size=3pt] at (1,0) {};

\draw (4,2) -- (3.5,1);
\draw (4,2) -- (4.5,1);

\draw (3.5,1) -- (3,0);
\draw (3.5,1) -- (4,0);

\draw[blue,dashed] (3,0) .. controls (3.5,-0.25) .. (4,0);
\draw[blue,dashed] (4,0) .. controls (4.25,0) .. (4.5,1);

\node[blue] at (3.75,-0.6) {$L_1$};

\node[circle, fill=black, inner sep=0, minimum size=3pt] at (4,2) {};
\node[circle, fill=black, inner sep=0, minimum size=3pt] at (3.5,1) {};
\node[circle, fill=black, inner sep=0, minimum size=3pt] at (4.5,1) {};
\node[circle, fill=black, inner sep=0, minimum size=3pt] at (3,0) {};
\node[circle, fill=black, inner sep=0, minimum size=3pt] at (4,0) {};

\draw (7,2) -- (6.5,1);
\draw (7,2) -- (7.5,1);

\draw (6.5,1) -- (6,0);
\draw (6.5,1) -- (7,0);

\draw[blue,dashed] (6,0) .. controls (6.5,-0.25) .. (7,0);
\draw[blue,dashed] (6,0) .. controls (6.5,-0.5) and (7.5,-0.5) .. (7.5,1);

\node[blue] at (6.75,-0.6) {$L_2$};

\node[circle, fill=black, inner sep=0, minimum size=3pt] at (7,2) {};
\node[circle, fill=black, inner sep=0, minimum size=3pt] at (6.5,1) {};
\node[circle, fill=black, inner sep=0, minimum size=3pt] at (7.5,1) {};
\node[circle, fill=black, inner sep=0, minimum size=3pt] at (6,0) {};
\node[circle, fill=black, inner sep=0, minimum size=3pt] at (7,0) {};
\end{tikzpicture}
\normalsize
\caption{A binary tree $T$ (left) and (drawn with dashed blue lines) the two possible leaf trees $L_1$ and $L_2$ that it induces (right). With the notation introduced in Definition~\ref{def:induced-leaf-tree}, $\L(T) = \set{L_1,L_2}$.}\label{fig:leaf-trees}
\end{figure}

We can now state the following central lemma:

\begin{lemma}[Number of induced leaf trees of binary trees]
\label{lem:num-of-induced-leaf-trees}
Let $l \in \N$, $T \in \mathfrak{B}_l$, and $d = (d_w)_{w \in \vleaf(T)}$ with $d_w \geq 1$ for all $w \in \vleaf(T)$. Then:
\begin{enumerate}[label=(\roman*)]
\item We have:
\begin{equation}
\label{eq:num-of-induced-leaf-trees-1}
\card{\mathfrak{L}(T)} = \prod_{v \in V(T) \setminus \set{*}} \card{\mathcal{L}_T(v)}.
\end{equation}
\item Define:
\begin{equation}
\mathfrak{L}(T, d)
= \set{L \in \mathfrak{L}(T) : d_L(w) = d_w\ \forall w \in \vleaf(T)}.
\end{equation}
Then, we either have $\card{\mathfrak{L}(T, d)} = 0$ or
\begin{equation}
\label{eq:num-of-induced-leaf-trees-2}
\card{\mathfrak{L}(T, d)}
= \frac{\prod_{v \in \vfork(T) \setminus \set{*}} \left( \sum_{w \in \mathcal{L}_T(v)} (d_w - 2) + 2 \right)}{\prod_{w \in \vleaf(T)} (d_w - 1)!}.
\end{equation}
\end{enumerate}
\end{lemma}

The proof of Lemma \ref{lem:num-of-induced-leaf-trees} is given in Appendix \ref{app:combinatorial-proofs}. We just want to explain the idea of the proof here using Figure \ref{fig:types-of-induction} because the inductive scheme of the proof differs in items (i) and (ii) of Lemma \ref{lem:num-of-induced-leaf-trees}.  To prove part (i), we use ``induction from the root'' (see Figure~\ref{fig:induction-root}). For this type of induction, the root $*$ of the initial binary tree $T$ is removed, splitting it into two new binary trees, $T_1$ and $T_2$. During this operation, the associated leaf tree $L$ is also split into two trees, $L_1$ and $L_2$ -- a process during which one edge gets lost. To prove part (ii), we use ``induction from the lowest leaves'' (see Figure~\ref{fig:induction-leaf}). In this case, the binary tree $T$ is transformed into a new tree $T'$ by removing the lowest-hanging ``cherry'' of two leaves. All edges of the leaf tree $L$ that previously terminated in the removed leaves are passed on to the parent vertex $p$ of the removed leaves. This process would yield one self-loop in the parent vertex $p$; this self-loop is omitted. That way, a new leaf tree $L'$ is created.

Now consider a binary tree $T \in \B$ and an induced leaf tree $L \in \L(T)$.
Like any graph, the leaf tree $L=(\vleaf(T),E(L))$ has a degree function \linebreak ${d_L:\ \vleaf(T) \to \N}$, assigning to each vertex a degree (see Definition~\ref{def:properties-of-graphs}).
Using the method from Figure~\ref{fig:induction-leaf} recursively, it is possible to define a ``degree'' $d_L(v)$ of $v$ in $L$ for all forks $v \in \vfork(T)$ of $T$. Take the vertex $p \in \vfork(T)$ from Figure~\ref{fig:induction-leaf} as an example. For this vertex, we have $d_L(p) = d_{L'}(p) = 2$ as can be seen from the right half of the figure.\footnote{We define a degree $d_L(p)$ of $p$ in $L$, although $p$ is not a vertex of $L$. This definition also depends on the associated binary tree $T$. Hence, the correct notation would be $d_L^{(T)}(p)$ instead of $d_L(p)$.} Let us return to general trees $T \in \B$ and $L \in \L(T)$ and to a general vertex $v \in \vfork(T)$. Some basic considerations using the handshaking lemma (Lemma~\ref{lem:handshaking-lemma}) yield the formula\footnote{The expression $e(L[\leavesof{T}{v}])$ denotes the number of edges of the subgraph of $L$ induced by the vertex set $\leavesof{T}{v}$; see Appendix \ref{app:graph-theory}.}
\begin{equation}
\begin{split}
d_L(v)
&= \Bigg( \sum_{w \in \leavesof{T}{v}} d_L(w) \Bigg) - 2 e(L[\leavesof{T}{v}]) \\
&= \Bigg( \sum_{w \in \leavesof{T}{v}} d_L(w) \Bigg) - 2 \left( \card{\leavesof{T}{v}} - 1 \right) \\
&= \sum_{w \in \leavesof{T}{v}} (d_L(w) - 2) + 2.
\end{split}
\end{equation}
This is exactly the expression found in the numerator of \eqref{eq:num-of-induced-leaf-trees-2}.

\begin{figure}
\begin{subfigure}{\textwidth}
\begin{center}
\tiny
\begin{tikzpicture}
\node[anchor=south] at (2,4.25) {$T$};

\draw (2,4) -- (1,3);
\draw (2,4) -- (3,3);

\draw[dotted,thick] (0.75,3) .. controls (1,3.25) .. (1.25,3);
\draw[dotted,thick] (2.75,3) .. controls (3,3.25) .. (3.25,3);

\draw (1,3) -- (0.5,2);
\draw (1,3) -- (1.5,2);
\draw (3,3) -- (2.5,2);
\draw (3,3) -- (3.5,2);

\draw (0.5,2) -- (0,1);
\draw (0.5,2) -- (1,1);

\draw[blue,dashed] (0,1) .. controls (0.5,0.75) .. (1,1);
\draw[blue,dashed] (1,1) .. controls (1.25,1) .. (1.5,2);
\draw[blue,dashed] (1,1) .. controls (1.75,0.5) .. (2.5,2);
\draw[blue,dashed] (2.5,2) .. controls (3,1.75) .. (3.5,2);

\node[blue] at (2,0.5) {$L$};

\node[circle, fill=black, inner sep=0, minimum size=3pt] at (2,4) {};
\node[circle, fill=black, inner sep=0, minimum size=3pt] at (1,3) {};
\node[circle, fill=black, inner sep=0, minimum size=3pt] at (3,3) {};
\node[circle, fill=black, inner sep=0, minimum size=3pt] at (0.5,2) {};
\node[circle, fill=black, inner sep=0, minimum size=3pt] at (1.5,2) {};
\node[circle, fill=black, inner sep=0, minimum size=3pt] at (2.5,2) {};
\node[circle, fill=black, inner sep=0, minimum size=3pt] at (3.5,2) {};
\node[circle, fill=black, inner sep=0, minimum size=3pt] at (0,1) {};
\node[circle, fill=black, inner sep=0, minimum size=3pt] at (1,1) {};

\draw[-stealth, line width=0.75mm] (4,2.5) -- (4.75,2.5);

\node[anchor=south] at (6,3.25) {$T_1$};
\node[anchor=south] at (8,3.25) {$T_2$};

\draw[dashed] (7,4) -- (6,3);
\draw[dashed] (7,4) -- (8,3);

\draw (6,3) -- (5.5,2);
\draw (6,3) -- (6.5,2);
\draw (8,3) -- (7.5,2);
\draw (8,3) -- (8.5,2);

\draw (5.5,2) -- (5,1);
\draw (5.5,2) -- (6,1);

\draw[blue,dashed] (5,1) .. controls (5.5,0.75) .. (6,1);
\draw[blue,dashed] (6,1) .. controls (6.25,1) .. (6.5,2);
\draw[blue,dashed] (7.5,2) .. controls (8,1.75) .. (8.5,2);

\node[blue] at (6,0.5) {$L_1$};
\node[blue] at (8,1.5) {$L_2$};

\node[circle, fill=black, inner sep=0, minimum size=3pt] at (6,3) {};
\node[circle, fill=black, inner sep=0, minimum size=3pt] at (8,3) {};
\node[circle, fill=black, inner sep=0, minimum size=3pt] at (5.5,2) {};
\node[circle, fill=black, inner sep=0, minimum size=3pt] at (6.5,2) {};
\node[circle, fill=black, inner sep=0, minimum size=3pt] at (7.5,2) {};
\node[circle, fill=black, inner sep=0, minimum size=3pt] at (8.5,2) {};
\node[circle, fill=black, inner sep=0, minimum size=3pt] at (5,1) {};
\node[circle, fill=black, inner sep=0, minimum size=3pt] at (6,1) {};
\end{tikzpicture}
\normalsize
\end{center}
\caption{Induction from the root.}
\label{fig:induction-root}
\end{subfigure}

\vspace{0.5cm}

\begin{subfigure}{\textwidth}
\begin{center}
\tiny
\begin{tikzpicture}
\node[anchor=south] at (2,4.25) {$T$};

\draw (2,4) -- (1,3);
\draw (2,4) -- (3,3);

\draw[dotted,thick] (0.75,3) .. controls (1,3.25) .. (1.25,3);
\draw[dotted,thick] (2.75,3) .. controls (3,3.25) .. (3.25,3);

\draw (1,3) -- (0.5,2);
\draw (1,3) -- (1.5,2);
\draw (3,3) -- (2.5,2);
\draw (3,3) -- (3.5,2);

\draw (0.5,2) -- (0,1);
\draw (0.5,2) -- (1,1);

\draw[blue,dashed] (0,1) .. controls (0.5,0.75) .. (1,1);
\draw[blue,dashed] (1,1) .. controls (1.25,1) .. (1.5,2);
\draw[blue,dashed] (1,1) .. controls (1.75,0.5) .. (2.5,2);
\draw[blue,dashed] (2.5,2) .. controls (3,1.75) .. (3.5,2);

\node[blue] at (2,0.5) {$L$};

\node[circle, fill=black, inner sep=0, minimum size=3pt] at (2,4) {};
\node[circle, fill=black, inner sep=0, minimum size=3pt] at (1,3) {};
\node[circle, fill=black, inner sep=0, minimum size=3pt] at (3,3) {};
\node[circle, fill=black, inner sep=0, minimum size=3pt] at (0.5,2) {};
\node[circle, fill=black, inner sep=0, minimum size=3pt] at (1.5,2) {};
\node[circle, fill=black, inner sep=0, minimum size=3pt] at (2.5,2) {};
\node[circle, fill=black, inner sep=0, minimum size=3pt] at (3.5,2) {};
\node[circle, fill=black, inner sep=0, minimum size=3pt] at (0,1) {};
\node[circle, fill=black, inner sep=0, minimum size=3pt] at (1,1) {};

\draw[-stealth, line width=0.75mm] (4,2.5) -- (4.75,2.5);

\node[anchor=south] at (6,3.25) {$T_1$};
\node[anchor=south] at (8,3.25) {$T_2$};

\draw[dashed,dashed] (7,4) -- (6,3);
\draw[dashed,dashed] (7,4) -- (8,3);

\draw (6,3) -- (5.5,2);
\draw (6,3) -- (6.5,2);
\draw (8,3) -- (7.5,2);
\draw (8,3) -- (8.5,2);

\draw (5.5,2) -- (5,1);
\draw (5.5,2) -- (6,1);

\draw[blue,dashed] (5,1) .. controls (5.5,0.75) .. (6,1);
\draw[blue,dashed] (6,1) .. controls (6.25,1) .. (6.5,2);
\draw[blue,dashed] (6,1) .. controls (6.4,0.75) .. (6.83,1);
\draw[blue,dashed] (7.17,1) .. controls (7.5,0.75) .. (7.5,2);
\draw[blue,dashed] (7.5,2) .. controls (8,1.75) .. (8.5,2);

\node[blue] at (6,0.5) {$L_1$};
\node[blue] at (8,1.5) {$L_2$};

\node[circle, fill=black, inner sep=0, minimum size=3pt] at (6,3) {};
\node[circle, fill=black, inner sep=0, minimum size=3pt] at (8,3) {};
\node[circle, fill=black, inner sep=0, minimum size=3pt] at (5.5,2) {};
\node[circle, fill=black, inner sep=0, minimum size=3pt] at (6.5,2) {};
\node[circle, fill=black, inner sep=0, minimum size=3pt] at (7.5,2) {};
\node[circle, fill=black, inner sep=0, minimum size=3pt] at (8.5,2) {};
\node[circle, fill=black, inner sep=0, minimum size=3pt] at (5,1) {};
\node[circle, fill=black, inner sep=0, minimum size=3pt] at (6,1) {};
\node[circle, fill=blue, inner sep=0, minimum size=3pt] at (6.83,1) {};
\node[circle, fill=blue, inner sep=0, minimum size=3pt] at (7.17,1) {};
\end{tikzpicture}
\normalsize
\end{center}
\caption{Induction from the root (on trees with extra vertices).}
\label{fig:induction-root-extra-vertices}
\end{subfigure}

\vspace{0.5cm}

\begin{subfigure}{\textwidth}
\begin{center}
\tiny
\begin{tikzpicture}
\node[anchor=south] at (2,4.25) {$T$};

\draw (2,4) -- (1,3);
\draw (2,4) -- (3,3);

\draw (1,3) -- (0.5,2);
\node[anchor=east] at (0.25,2) {$p$};
\draw (1,3) -- (1.5,2);
\draw (3,3) -- (2.5,2);
\draw (3,3) -- (3.5,2);

\draw[dotted,thick] (0.25,2) .. controls (0.5,2.25) .. (0.75,2);

\draw (0.5,2) -- (0,1);
\draw (0.5,2) -- (1,1);

\draw[blue,dashed] (0,1) .. controls (0.5,0.75) .. (1,1);
\draw[blue,dashed] (1,1) .. controls (1.25,1) .. (1.5,2);
\draw[blue,dashed] (1,1) .. controls (1.75,0.5) .. (2.5,2);
\draw[blue,dashed] (2.5,2) .. controls (3,1.75) .. (3.5,2);

\node[blue] at (2,0.5) {$L$};

\node[circle, fill=black, inner sep=0, minimum size=3pt] at (2,4) {};
\node[circle, fill=black, inner sep=0, minimum size=3pt] at (1,3) {};
\node[circle, fill=black, inner sep=0, minimum size=3pt] at (3,3) {};
\node[circle, fill=black, inner sep=0, minimum size=3pt] at (0.5,2) {};
\node[circle, fill=black, inner sep=0, minimum size=3pt] at (1.5,2) {};
\node[circle, fill=black, inner sep=0, minimum size=3pt] at (2.5,2) {};
\node[circle, fill=black, inner sep=0, minimum size=3pt] at (3.5,2) {};
\node[circle, fill=black, inner sep=0, minimum size=3pt] at (0,1) {};
\node[circle, fill=black, inner sep=0, minimum size=3pt] at (1,1) {};

\draw[-stealth, line width=0.75mm] (4,2.5) -- (4.75,2.5);

\node[anchor=south] at (7,4.25) {$T'$};

\draw (7,4) -- (6,3);
\draw (7,4) -- (8,3);

\draw (6,3) -- (5.5,2) node[anchor=east] {$p$};
\draw (6,3) -- (6.5,2);
\draw (8,3) -- (7.5,2);
\draw (8,3) -- (8.5,2);

\draw[blue,dashed] (5.5,2) .. controls (6,1.5) .. (7.5,2);
\draw[blue,dashed] (5.5,2) .. controls (6,1.75) .. (6.5,2);
\draw[blue,dashed] (7.5,2) .. controls (8,1.75) .. (8.5,2);

\node[blue] at (7,1.25) {$L'$};

\node[circle, fill=black, inner sep=0, minimum size=3pt] at (7,4) {};
\node[circle, fill=black, inner sep=0, minimum size=3pt] at (6,3) {};
\node[circle, fill=black, inner sep=0, minimum size=3pt] at (8,3) {};
\node[circle, fill=black, inner sep=0, minimum size=3pt] at (5.5,2) {};
\node[circle, fill=black, inner sep=0, minimum size=3pt] at (6.5,2) {};
\node[circle, fill=black, inner sep=0, minimum size=3pt] at (7.5,2) {};
\node[circle, fill=black, inner sep=0, minimum size=3pt] at (8.5,2) {};
\end{tikzpicture}
\normalsize
\end{center}
\caption{Induction from the lowest leaves.}
\label{fig:induction-leaf}
\end{subfigure}
\caption{Schematic representation of different types of induction steps on leaf trees and their associated binary tree.}
\label{fig:types-of-induction}
\end{figure}
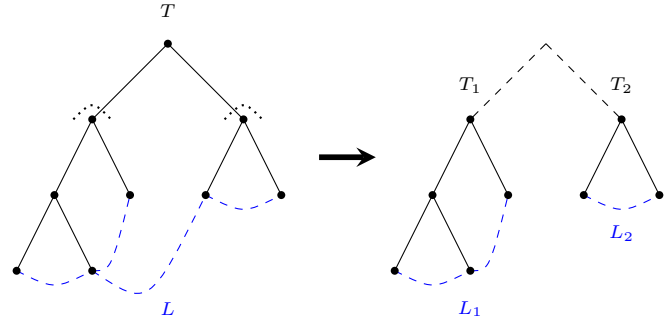
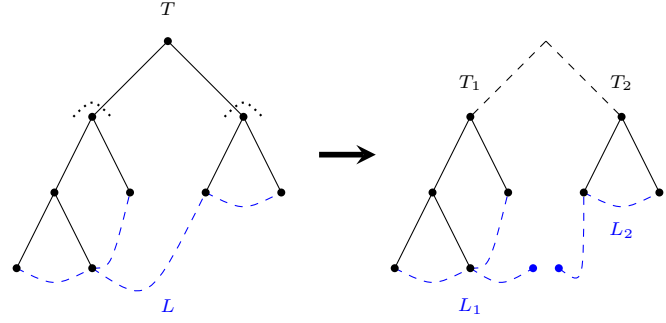
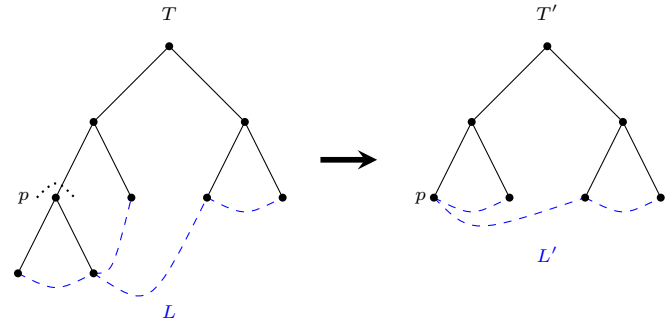

\bigskip
In the previous lemma, we have found an exact expression for $\card{\mathfrak{L}(T,d)}$, i.e. for the number of induced leaf trees given a certain degree sequence. However, later we will need a bound for $\card{\mathfrak{L}(T,d)}$ that is of a slightly different form:

\begin{lemma}
\label{lem:num-of-leaf-trees-bound}
Let $l \in \N$, $T \in \mathfrak{B}_l$, and $d = (d_w)_{w \in \vleaf(T)}$ with $d_w \geq 1$ for all $w \in \vleaf(T)$. Then
\begin{equation}\label{eq:num-of-leaf-trees-bound}
\card{\mathfrak{L}(T, d)}
\leq 
\begin{cases}
\frac{\prod_{v \in V(T) \setminus \set{*}} \card{\mathcal{L}_T(v)}}{\prod_{w \in \vleaf(T)} (d_w - 1)!}  & \text{ if } \sum\limits_{w \in \vleaf(T)} d_w = 2(l-1) \\
0 & \text{ otherwise. }
\end{cases}
\end{equation}
\end{lemma}

The proof of this lemma can be found in Appendix~\ref{app:combinatorial-proofs}, too. The proof uses a slightly more general setting, namely leaf trees with extra vertices. The induction step is similar to the one shown in Figure~\ref{fig:induction-root} with the difference that the edge formerly connecting $L_1$ and $L_2$ is not lost, but split into two edges connecting to extra vertices. This is shown in Figure~\ref{fig:induction-root-extra-vertices}.

As a final step of preparation, we define a mapping $\mathfrak{B} \times \R_{\geq 0} \to \R_{\geq 0}$ that will become important later on and investigate it combinatorially.
\randbem{\vglue1cm suggestion: denote it as $\int_0^r d^T s$.}

\begin{definition}
Let
\begin{equation}
\int:\ \mathfrak{B} \times \R_{\geq 0} \to \R_{\geq 0},\ (T,r) \mapsto \int_0^r{T} = \overset{\rightarrow}{\prod_{v \in \vfork(T)}} \int_0^{s_{p_T(v)}} \dd{s_v}\; \frac{1}{2},
\end{equation}
where the product is ordered according to the total order on $V(T)$ defined in Definition~\ref{def:binary-tree} and where $s_{p_T(*)} = s_\dagger = r$.
\end{definition}

\begin{lemma}
\label{lem:bound-on-tree-integral}
For $T \in \mathfrak{B}$ and $r \in \R_{\geq 0}$, we have:
\begin{equation}
\int_0^r{T}
\leq r^{l(T) - 1} \cdot \prod_{v \in V(T)} (\card{\mathcal{L}_T(v)})^{-1}.
\end{equation}
\end{lemma}

\begin{proof}
Let $T \in \mathfrak{B}$ and $r \in \R_{\geq 0}$. We prove the lemma using induction on $l(T)$. For $l(T) = 1$, the statement is clear. Now, let $l(T) \geq 2$. After removing the root $*$ from $T$, it breaks apart into two binary trees $T_1$ and $T_2$ rooted at $c_T^{(1)}(*)$ and $c_T^{(2)}(*)$ respectively\footnote{This stretches our definition of a binary tree a bit as we required binary trees to be rooted at $(0,0)$ above, but this detail does not matter for this proof.} (see Figure~\ref{fig:induction-root}). Then by definition   
\beq
\int_0^r{T}
= \int_0^r \dd{s_*} \; \frac{1}{2} \left( \int_0^{s_*}{T_1} \right) \left( \int_0^{s_*}{T_2} \right) \;.
\eeq
The inductive hypothesis applies to both factors on the right hand side, yielding a power
\beq
\card{\mathcal{L}_T\left(c_T^{(1)}(*)\right)} - 1 \;  + \; \card{\mathcal{L}_T\left(c_T^{(2)}(*)\right)} - 1
=
l(T) -2
\eeq
of $s_*$ and a factor
\beq
\prod_{v \in V(T_1) \cup V(T_2)} (\card{\mathcal{L}_T(v)})^{-1} 
=
\prod_{v \in V(T) \setminus \set{*}} (\card{\mathcal{L}_T(v)})^{-1}
\eeq
so that
\beq
\int_0^r{T}
= 
\sfrac{1}{2}  \int_0^r s_*^{l(T)-2} \; \dd{s_*} \;  \prod_{v \in V(T) \setminus \set{*}} (\card{\mathcal{L}_T(v)})^{-1}
\eeq
The integral over $s_*$ gives
\beq
\sfrac{1}{2}  \int_0^r s_*^{l(T)-2} \; \dd{s_*} 
=
\frac{r^{l(T)-1}}{2 (l(T) -1)} 
\le
\frac{r^{l(T)-1}}{l(T)}
\eeq
Noting that $l(T) = \card{\mathcal{L}_T(*)}$, this completes the proof. 
\end{proof}

\section{Proof of convergence}
\label{sec:proof-of-convergence}

In this section, we prove the convergence Theorem~\ref{thm:convergence-of-formal-solution} by building on Section~\ref{sec:combinatorics-of-graphs}. To make the presentation more self-containted, we first review the role of determinants in fermionic quantum field theory and the technique of ``trimming the tree''.

\subsection{Role of determinants in fermionic QFT}
\label{subsec:determinants}

The following lemma shows how determinants naturally arise when applying exponentiated laplacians to monomials.

\begin{lemma}
\label{lem:laplacians-and-determinants}
Let $D \subseteq \mathcal{X}$ and $\Mcov \in M_N(\R)$. Then, there exist $\varepsilon_D^{A,A'} \in \set{\pm 1}$ for all $A,A' \subseteq D$ such that
\begin{equation}
e^{\Delta_\Mcov} \prod_{X \in D} \psi(X)
= \sum_{\substack{A,A' \subseteq D \\ \card{A}=\card{A'} \\ A \cap A' = \emptyset}} \varepsilon_D^{A,A'} \det(\Mcov_{A,A'}) \prod_{X \in D \setminus (A \cup A')} \psi(X).
\end{equation}
Here, $\Mcov_{A,A'}$ denotes a quadratic minor of the matrix $\Mcov$.
\end{lemma}

This is an adapted version of \cite[Lemma~5]{SW}. Although the proof remains similar, we give it explicitly in Appendix~\ref{app:selected-proofs}.

The preceding lemma shows the importance of bounding determinants. The basis for these bounds is Gram's inequality:

\begin{lemma}[Gram's inequality {\cite[Lemma~B.13]{msbook}}]
\label{lem:grams-inequality}
Let $\mathcal{H}$ be a Hilbert space with inner product $\langle \cdot, \cdot \rangle$, and for $f=(f_1,\dots,f_n)\in\mathcal{H}^n$ and $g=(g_1,\dots,g_n)\in\mathcal{H}^n$, let
\begin{equation}
(f,g) = \det((\langle f_i, g_j \rangle)_{1 \leq i,j \leq n}).
\end{equation}
Then:
\begin{equation}
|(f,g)| \leq \prod_{i=1}^n \lVert f_i \rVert \cdot \lVert g_i \rVert.
\end{equation}
\end{lemma}

The proof can be found in \cite[Appendix~B]{msbook}.
%
Next, we define what we mean by a matrix to have a certain Gram constant:

\begin{definition}[Gram constant]
\label{def:gram-constant}
Let $M \in M_n(\R)$ be a matrix and $\gamma \in \R_{>0}$. We say that $M$ has a \textit{Gram constant} $\gamma$ iff there exists a Hilbert space $\mathcal{H}$ and $f = (f_1, \dots, f_n) \in \mathcal{H}^n$, $g = (g_1, \dots, g_n) \in \mathcal{H}^n$ such that
\begin{equation}
M = (\langle f_i, g_j \rangle)_{1 \leq i,j \leq n}
\end{equation}
and $\lVert f_i \rVert \leq \gamma$ and $\lVert g_i \rVert \leq \gamma$ for all $i \in [n]$. We call $(\mathcal{H},f,g)$ a \textit{Gram decomposition} of $M$ corresponding to the Gram constant $\gamma$.
\end{definition}

With this definition, the following statement follows directly from Lemma~\ref{lem:grams-inequality}:

\begin{corollary}
\label{cor:grams-inequality}
Let $M \in M_n(\R)$ be a matrix with Gram constant $\gamma \in \R_{>0}$. Let $A,A' \subseteq [n]$ with $\card{A} = \card{A'} = k$. Then:
\begin{equation}
|\det(M_{A,A'})| \leq \gamma^{2k}.
\end{equation}
In particular:
\begin{equation}
|\det(M)| \leq \gamma^{2n}.
\end{equation}
\end{corollary}

Next, we consider tensor products of matrices.

\begin{lemma}
\label{lem:gram-constant-and-tensor-product}
Let $M \in M_n(\R)$ and $M' \in M_{n'}(\R)$ be two matrices with Gram constants $\gamma$ and $\gamma'$ respectively. Then, the matrix $M \otimes M'$ has Gram constant $\gamma \cdot \gamma'$.
\end{lemma}

\begin{proof}
Let $(\mathcal{H},f,g)$ and $(\mathcal{H}',f',g')$ be Gram decompositions of $M$ and $M'$ corresponding to the Gram constants $\gamma$ and $\gamma'$ respectively. Then, $(\mathcal{H} \otimes \mathcal{H}', f \otimes f', g \otimes g')$ is a Gram decomposition with Gram constant $\gamma \cdot \gamma'$.
\end{proof}

Lastly, we recall how one can obtain a Gram constant for positive matrices; 
this is where positivity comes into play crucially.

\begin{lemma}[{\cite[Lemma~7]{msbook}}]
\label{lem:gram-constant-of-positive-matrices}
Let $M \in M_n(\R)$ be a positive semidefinite matrix. Then, $M$ has a Gram constant
\begin{equation}
\gamma = \max\set{\sqrt{M_{ii}} : i \in [n]}.
\end{equation}
\end{lemma}

\begin{proof}
This proof follows the proof of \cite[Lemma~7]{msbook}.
Since $M \geq 0$, there is a matrix $B \in M_n(\R)$ with $B \geq 0$ such that $M = B^2$. If $b_i = (B_{ij})_{j \in [n]}$ is the $i$-th row vector of $B$, this means that $M_{ij} = \langle b_i, b_j \rangle$ for all $i,j \in [n]$. Hence, we have found a Gram decomposition of $M$ with Gram constant
\begin{equation}
\gamma = \max\set{\lVert b_i \rVert : i \in [n]} = \max\set{\sqrt{M_{ii}} : i \in [n]}.
\end{equation}
\end{proof}

While every matrix has some Gram representation, the point of Definition \ref{def:gram-constant} is that the constant $\gamma$ should be independent of the size $n$ of the matrix, so that bounds can be uniform in $n$. In applications to quantum many-fermion systems, it is important to use a more general notion of a determinant bound $\delta_M$ that also applies directly to time-ordered covariances, for which the standard Gram bound diverges in the time-continuum limit \cite{PSUV}. The estimates we give in the present paper, in particular in the proof of Theorem \ref{thm:convergence-of-formal-solution}, remain unchanged when the Gram constant is replaced by the determinant bound, so Theorem \ref{thm:convergence-of-formal-solution} also holds in this more general context.

\subsection{The technique of ``trimming the tree''}
\label{subsec:trimming-the-tree}

An important technique when it comes to bounding tree expansions is that of ``trimming the tree''. The following lemma specifies what is meant by this.

\begin{lemma}[``Trimming the tree'' {\cite[Appendix C]{brydgesShortCourseCluster}}]
\label{lem:trimming-the-tree}
Let $C \in \GL_N(\R)$ be a skew-symmetric matrix such that
\begin{equation}
|C| = \sup_{X \in \X} \int \dd{Y}\; |C(X,Y)| < \infty.
\end{equation}
Let $v_m \in \R^{\X^{m}}$ be an antisymmetric tensor for all $m \in \N$.
Let $p \in \N$, let $T \in \T_p$ be a tree, and let $m_1, \dots, m_p \geq 1$ such that $m_q \geq d_T(q)$ for all $q \in [p]$. Define $D=\set{(q,i) : q \in [p],\ i \in [m_q]}$. Then:
\begin{equation}
\label{eq:trimming-the-tree}
\begin{split}
&\max_{(q',i') \in D} \sup_{\tilde X \in \X} \int \Bigg( \prod_{(q,i) \in D \setminus \set{(q',i')}} \dd{X_{(q,i)}} \Bigg) \Bigg( \prod_{q=1}^p \abs{v_{m_q}\Big(X_{(q,1)}, \dots, X_{(q,m_q)}\Big)} \Bigg) \\
&\quad \cdot \Bigg( \prod_{\set{q,q'} \in E(T)} \abs{C\Big(X_{(q,i_{T,q}(q'))}, X_{(q',i_{T,q'}(q))}\Big)} \Bigg) \Bigg|_{X_{(q',i')} = \tilde X} \\
&\leq \abs{C}^{p-1} \prod_{q=1}^p \abs{v_{m_q}},
\end{split}
\end{equation}
where we use Notation~\ref{not:enumeration-of-neighbours} and the norm $\abs{v_{m_q}}$ from Definition~\ref{def:h-norm}.
\end{lemma}

The proof of this lemma can be found in \cite[Appendix C]{brydgesShortCourseCluster}. Due to slightly different notation in this reference, we also give the proof in Appendix~\ref{app:selected-proofs} for the reader's convenience.

\subsection{Proof of Theorem~\ref{thm:convergence-of-formal-solution}} We are now finally ready to prove the central theorem of this paper.

\begin{proof}[Proof of Theorem~\ref{thm:convergence-of-formal-solution}]
Consider the formal solution from Theorem~\ref{thm:formal-solution}. By writing out every $\G_0(s_{p(v)},\Psi^{(w)}) = e^{\frac{1}{2} \Delta_{C_{s_{p(v)},0}}} \V(\Psi^{(w)})$ and then re-ordering the terms (noting that all Laplacians commute with one another), it can be rewritten as
\begin{equation}
\begin{split}
&\G_l(t=1) = \sum_{T \in \B_l} \left(\ordprod{v \in \vfork(T)} \int_0^{s_{p(v)}} \dd{s_v} \; \left(-\frac{1}{2}\right)\right) \\
&\qquad \left\llbracket \exp(\frac{1}{2} \sum_{v \in \vfork(T)} (s_{p(v)} - s_v) \Delta_C^{(\leavesof{T}{v})} + \frac{1}{2} \sum_{v \in \vleaf(T)} \Delta_C^{(v)}) \right. \\
&\qquad \cdot \left. \left( \prod_{v \in \vfork(T)} \innerprod{\nabla^{(\leavesof{T}{c_1(T)})}}{C \nabla^{(\leavesof{T}{c_2(T)})}} \right) \cdot \left(\prod_{v \in \vleaf(T)} \V\left(\Psi^{(v)}\right)\right) \right\rrbracket.
\end{split}
\end{equation}
We can further rewrite
\begin{equation}
\begin{split}
\frac{1}{2} \sum_{v \in \vfork(T)} (s_{p(v)} - s_v) \Delta_C^{(\leavesof{T}{v})} + \frac{1}{2} \sum_{v \in \vleaf(T)} \Delta_C^{(v)}
&= \Delta_C[M(T,\underline{s})],
\end{split}
\end{equation}
where $M(T,\underline{s}) \in M_l(\R)$ is a matrix indexed by $\vleaf(T)$ and defined by
\begin{equation}
M(T,\underline{s}) = \frac{1}{2}\Bigg(\bbbone + \sum_{v \in \vfork(T)} \underbrace{(s_{p(v)} - s_v)}_{\geq 0} J_{\leavesof{T}{v}} \Bigg).
\end{equation}
In this definition, we use the matrix $J_I \in M_l(\R)$ for $I \subseteq \vleaf(T)$ given by
\begin{equation}
(J_I)_{ij} = \begin{cases}
1 & \text{if $i\in I$ and $j \in I$} \\
0 & \text{else}
\end{cases}.
\end{equation}
We note that $M(T,\underline{s})$ is a sum of positive semidefinite matrices, hence it is itself positive semidefinite. Additionally, we have $M_{vv}(T,\underline{s}) = \frac{1}{2}(2 - s_{p(v)}) \le 1$ for all $v \in \vleaf(T)$. Hence, by Lemma~\ref{lem:gram-constant-of-positive-matrices}, $M(T,\underline{s})$ has a Gram constant of $1$.

Write the initial condition $\G(0,\Psi) = \V(\Psi)$ as
\begin{equation}
\V(\Psi) = \sum_{m \geq 0} \int \dd^m{\underline{X}}\ v_m(\underline{X}) \psi^m(\underline{X}).
\end{equation}
Using the above rewriting of $\G_l(t=1)$, we can now write:\footnote{For better legibility, we write $\norm{.}_h$ for $\norm{\llbracket . \rrbracket}_h$. Additionally, for $w \in \vleaf(T)$, $m \in \N$ and $\underline{X} \in \X^m$, we write $\psi_{(w)}^m(\underline{X})$ instead of $(\psi^{(w)})^m(\underline{X})$.}
\begin{equation}
\begin{split}
\norm{\G_l(t=1)}_h
&\leq \sum_{T \in \B_l} \sum_{\substack{m_w \in \N_0 \\ \forall w \in \vleaf(T)}} \left(\ordprod{v \in \vfork(T)} \int_0^{s_{p(v)}} \dd{s_v} \; \frac{1}{2}\right) \\
&\qquad \cdot \left\lVert e^{\Delta_C[M(T,\underline{s})]} \left(\prod_{v \in \vfork(T)} \innerprod{\nabla^{(\leavesof{T}{c_1(T)})}}{C \nabla^{(\leavesof{T}{c_2(T)})}}\right) \right. \\
&\qquad\quad \left. \int \prod_{w \in \vleaf(T)} \dd^{m_w}{\underline{X_w}}\ v_{m_w}(\underline{X_w}) \cdot \psi_{(w)}^{m_w}(\underline{X_w}) \right\rVert_h.
\end{split}
\end{equation}
Consider the operator $\prod_{v \in \vfork(T)} \innerprod{\nabla^{(\leavesof{T}{c_1(T)})}}{C \nabla^{(\leavesof{T}{c_2(T)})}}$; for every $v \in \vfork(T)$, it picks out $w_1 \in \leavesof{T}{c_1(v)}$ and $w_2 \in \leavesof{T}{c_2(v)}$ and ``connects'' them with a $C\left(\left(\underline{X_{w_1}}\right)_i, \left(\underline{X_{w_2}}\right)_j\right)$ (with $i \in [m_{w_1}]$ and $j \in [m_{w_2}]$). One quickly notices that this produces a sum over 
induced leaf trees (see Definition~\ref{def:induced-leaf-tree}). Hence, we can rewrite\footnote{For $\underline{X} \in \X^m$ and $I=\set{n_1, \dots, n_k}$ with $n_1 < \dots < n_k$, use the notation $\left(\underline{X}\right)_I = (X_{n_1}, \dots, X_{n_k}) \in \X^k$.}
\begin{equation}
\begin{split}
&\norm{\G_l(t=1)}_h
\leq \sum_{T \in \B_l} \sum_{\substack{m_w \in \N_0 \\ \forall w \in \vleaf(T)}} \sum_{L \in \L(T)} \left(\ordprod{v \in \vfork(T)} \int_0^{s_{p(v)}} \dd{s_v} \; \frac{1}{2}\right) \\
&\qquad \cdot \Bigg\lVert e^{\Delta_C[M(T,\underline{s})]} \sum_{\substack{\iota = (\iota_w)_{w \in \vleaf(T)} \\ \iota_w:\ N_L(w) \to [m_w] \\ \text{injective}}} \varepsilon_{T,L,\iota} \int \Bigg(\prod_{w \in \vleaf(T)} \dd^{m_w}{\underline{X_w}}\ v_{m_w}(\underline{X_w}) \\
&\qquad\qquad \cdot \psi_{(w)}^{m_w - d_L(w)}\left(\left(\underline{X_w}\right)_{[m_w] \setminus \set{\iota_w(w') : w' \in N_L(w)}}\right) \Bigg) \\
&\qquad\quad \cdot \Bigg( \prod_{\set{w,w'} \in E(L)} C\left(\left(\underline{X_w}\right)_{\iota_w(w')}, \left(\underline{X_{w'}}\right)_{\iota_{w'}(w)}\right) \Bigg) \Bigg\rVert_h,
\end{split}
\end{equation}
where $\varepsilon_{T,L,\iota} \in \set{\pm 1}$. Since the tensor $v_{m_w}$ is antisymmetric, the choice of $\set{\iota_w(w') : w' \in N_L(w)}$ in $[m_w]$ does not matter; we can always choose $\set{\iota_w(w') : w' \in N_L(w)} = \set{1, \dots, d_L(w)}$ with some fixed ordering of elements. We just have to take into consideration the combinatorial factor
\begin{equation}\label{eq:derivcomb}
\sum_{\substack{\iota = (\iota_w)_{w \in \vleaf(T)} \\ \iota_w:\ N_L(w) \to [m_w] \\ \text{injective}}} 1 = \prod_{w \in \vleaf(T)} \binom{m_w}{d_L(w)} d_L(w)!.
\end{equation}
Using this, we get
\begin{equation}
\begin{split}
&\norm{\G_l(t=1)}_h
\leq \sum_{T \in \B_l} \sum_{\substack{m_w \in \N_0 \\ \forall w \in \vleaf(T)}} \sum_{L \in \L(T)} \left(\ordprod{v \in \vfork(T)} \int_0^{s_{p(v)}} \dd{s_v} \; \frac{1}{2}\right) \\
&\qquad \cdot \left(\prod_{w\in\vleaf(T)} \binom{m_w}{d_L(w)} d_L(w)! \right) \\
&\qquad \cdot \Bigg\lVert e^{\Delta_C[M(T,\underline{s})]} \int \Bigg( \prod_{w \in \vleaf(T)} \dd^{d_L(w)}{\underline{X_w}}\ \dd^{m_w - d_L(w)}{\underline{Y_w}}\ v_{m_w}(\underline{X_w},\underline{Y_w}) \\
&\qquad\quad \cdot \psi_{(w)}^{m_w - d_L(w)}(\underline{Y_w}) \Bigg) \cdot \Bigg(\prod_{\set{w,w'} \in E(L)} C\left(\left(\underline{X_w}\right)_{\iota_{L,w}(w')}, \left(\underline{X_{w'}}\right)_{\iota_{L,w'}(w)}\right) \Bigg) \Bigg\rVert_h,
\end{split}
\end{equation}
where we use Notation~\ref{not:enumeration-of-neighbours}.

Next, we need to deal with the term $e^{\Delta_C[M(T,\underline{s})]} \prod_{w \in \vleaf(T)} \psi_{(w)}^{m_w - d_L(w)}(\underline{Y_w})$. With a slight abuse of notation, we can write $\Delta_C[M(T,\underline{s})] = \Delta_{\Xi(T,\underline{s})}$ with $\Xi(T,\underline{s}) = M(T,\underline{s}) \otimes C$. Using Lemma~\ref{lem:gram-constant-and-tensor-product} and the above findings on $M(T,\underline{s})$, we conclude that $\Xi(T,\underline{s})$ has Gram constant $\gamma_C$. Using Lemma~\ref{lem:laplacians-and-determinants}, we find that
\begin{equation}
\begin{split}
&e^{\Delta_C[M(T,\underline{s})]} \prod_{w \in \vleaf(T)} \psi_{(w)}^{m_w - d_L(w)}(\underline{Y_w}) \\
=& \sum_{\substack{A_1,A_2 \subseteq D \\ \card{A_1}=\card{A_2} \\ A_1 \cap A_2 = \emptyset}} \varepsilon_D^{A_1,A_2} \det(\Xi_{A_1,A_2}) \prod_{(w,Y) \in D \setminus (A_1 \cup A_2)} \psi^{(w)}(Y),
\end{split}
\end{equation}
where $D = \set{(w,(\underline{Y_w})_i) : w \in \vleaf(T),\ i \in [m_w - d_L(w)]}$, and $\varepsilon_D^{A_1,A_2} \in \set{\pm 1}$ for all ${A_1,A_2 \subseteq D}$. We can rewrite this expression like so:
\begin{equation}
\begin{split}
&e^{\Delta_C[M(T,\underline{s})]} \prod_{w \in \vleaf(T)} \psi_{(w)}^{m_w - d_L(w)}(\underline{Y_w}) \\
=& \sum_{k=0}^\infty \sum_{\substack{A \subseteq D \\ \card{A}=2k}} \sum_{\substack{A_1 \subseteq A \\ \card{A_1}=k \\ A_2=A\setminus A_1}} \varepsilon_D^{A_1,A_2} \det(\Xi_{A_1,A_2}) \prod_{(w,Y) \in D \setminus A} \psi^{(w)}(Y) \\
=& \sum_{k=0}^\infty \sum_{\substack{a_w \in \N_0 \\ \forall w \in \vleaf(T) \\ \sum_w a_w = 2k}} \sum_{\substack{A_w \subseteq D_w \\ \card{A_w}=a_w \\ \forall w \in \vleaf(T) \\ A=\bigcup_w A_w}} \sum_{\substack{A_1 \subseteq A \\ \card{A_1}=k \\ A_2=A\setminus A_1}} \varepsilon_D^{A_1,A_2} \det(\Xi_{A_1,A_2}) \prod_{(w,Y) \in D \setminus A} \psi^{(w)}(Y),
\end{split}
\end{equation}
where for every $w \in \vleaf(T)$, we define
\begin{equation}
D_w = \set{(w',Y) \in D : w=w'} = \set{(w,(\underline{Y_w})_i):i\in[m_w-d_L(w)]}.
\end{equation}
Using Corollary~\ref{cor:grams-inequality}, we get $\abs{\det(\Xi_{A_1,A_2})} \leq \gamma_C^{2k} = \prod_{w \in \vleaf(T)} \gamma_C^{a_w}$. We have combinatorial factors
\begin{equation}
\sum_{\substack{A_w \subseteq D_w \\ \card{A_w}=a_w \\ \forall w \in \vleaf(T) \\ A=\bigcup_w A_w}} 1 = \prod_{w \in \vleaf(T)} \binom{m_w - d_L(w)}{a_w}
\text{ and }
\sum_{\substack{A_1 \subseteq A \\ \card{A_1}=k \\ A_2=A\setminus A_1}} 1 = \binom{2k}{k}.
\end{equation}
Using all this, we get after ``trimming the tree'' (see Lemma~\ref{lem:trimming-the-tree}):
\begin{equation}
\label{eq:convergence-proof-bound-1}
\begin{split}
&\norm{\G_l(t=1)}_h \\
&\leq (\omega_C \gamma_C^2)^{l-1} \sum_{T \in \B_l} \sum_{\substack{m_w \in \N_0 \\ \forall w \in \vleaf(T)}} \sum_{L \in \L(T)} \sum_{k=0}^\infty \sum_{\substack{a_w \in \N_0 \\ \forall w \in \vleaf(T) \\ \sum_w a_w = 2k}}
\left(\ordprod{v \in \vfork(T)} \int_0^{s_{p(v)}} \dd{s_v} \; \frac{1}{2}\right) \\
&\qquad \cdot \left(\prod_{w\in\vleaf(T)} \binom{m_w}{d_L(w)} d_L(w)! \cdot \binom{m_w - d_L(w)}{a_w} \gamma_C^{a_w} \right)  \cdot \binom{2k}{k} \\
&\qquad \cdot \prod_{w \in \vleaf(T)} \abs{v_{m_w}} \cdot h^{m_w - d_L(w) - a_w}.
\end{split}
\end{equation}
Because the summand no longer depends on $\underline{s}$, we can use Lemma~\ref{lem:bound-on-tree-integral} to bound the $\underline{s}$-integral by 
\begin{equation}
\label{eq:convergence-proof-integral-bound}
\ordprod{v \in \vfork(T)} \int_0^{s_{p(v)}} \dd{s_v} \; \frac{1}{2} = \int_0^1 T \leq \prod_{v \in V(T)} (\card{\leavesof{T}{v}})^{-1}.
\end{equation}
Next, we can rewrite
\begin{equation}
\sum_{L\in \L(T)} \mapsto \sum_{\substack{d_w \in \N_0 \\ \forall w\in\vleaf(T)}} \sum_{L\in\L(T,(d_w)_w)}
\end{equation}
and replace every $d_L(w)$ by $d_w$. After this rewriting, no dependence on the tree $L$ is left. Hence, we can use Corollary~\ref{lem:num-of-leaf-trees-bound} and have
\begin{equation}
\sum_{L \in \L(T,d)} 1 \leq \frac{\prod_{v \in V(T) \setminus \set{*}} \card{\leavesof{T}{V}}}{\prod_{v \in \vleaf(T)}{(d_w-1)!}} \cdot \bbbone\Bigg(\sum_{w \in \vleaf(T)} d_w = 2(l-1)\Bigg)
\end{equation}
where $\bbbone(P) =1$ if $P$ is true and $0$ otherwise.
Note that the numerator of this fraction is cancelled by the expression in Equation~\eqref{eq:convergence-proof-integral-bound}. In total, we have
\begin{equation}
\label{eq:convergence-proof-bound-2}
\begin{split}
&\norm{\G_l(t=1)}_h
\leq (\omega_C \gamma_C^2)^{l-1} \sum_{T \in \B_l} \sum_{\substack{m_w \in \N_0 \\ \forall w \in \vleaf(T)}} \sum_{\substack{d_w \in \N_0 \\ \forall w\in\vleaf(T)}} \sum_{k=0}^\infty \sum_{\substack{a_w \in \N_0 \\ \forall w \in \vleaf(T) \\ \sum_w a_w = 2k}} \\
&\quad \cdot \Bigg( \prod_{w \in \vleaf(T)} \frac{d_w!}{(d_w-1)!} \cdot \bbbone\bigg( \sum_{w \in \vleaf(T)} d_w = 2(l-1) \bigg) \Bigg) \\
&\quad \cdot \left(\prod_{w\in\vleaf(T)} \binom{m_w}{d_w} \binom{m_w - d_w}{a_w} \gamma_C^{a_w} h^{m_w - d_w - a_w} \right)  \cdot \binom{2k}{k} \cdot \prod_{w \in \vleaf(T)} \abs{v_{m_w}}.
\end{split}
\end{equation}
By the tree identity, $\gamma_C^{2(l-1)} = \prod_{w \in \vleaf(T)} \gamma_C^{d_w}$. Moreover, the summand depends on $T$ only via $V_{\rm leaf} (T)$, so we can use Lemma~\ref{lem:num-of-binary-trees}:
\begin{equation}
\sum_{T \in \B_l} 1 = \frac{1}{l} \binom{2(l-1)}{l-1} \leq \frac{1}{l} (2e)^{l-1} < (2e)^{l-1}.
\end{equation}
Additionally, we can use the bound $\binom{2k}{k} \leq (2e)^k = \prod_{w \in \vleaf(T)} \left(\sqrt{2e}\right)^{a_w}$ on Equation~\eqref{eq:convergence-proof-bound-2}. Lastly, we can use the arithmetic-geometric inequality to obtain:
\begin{equation}
\begin{split}
&\prod_{w \in \vleaf(T)} \frac{d_w!}{(d_w-1)!} \cdot \bbbone\bigg( \sum_{w \in \vleaf(T)} d_w = 2(l-1) \bigg) \\
=& \prod_{w \in \vleaf(T)} d_w \cdot \bbbone\bigg( \sum_{w \in \vleaf(T)} d_w = 2(l-1) \bigg) \\
\leq& \left(\frac{1}{l} \cdot 2(l-1)\right)^l
\leq 2^l.
\end{split}
\end{equation}
This gives us:
\begin{equation}
\begin{split}
&\norm{\G_l(t=1)}_h
\leq 2 \cdot (4e \cdot \omega_C)^{l-1} \sum_{m_1, \dots, m_l = 0}^\infty \sum_{d_1, \dots, d_l = 0}^\infty \sum_{a_1, \dots, a_l = 0}^\infty \\
&\qquad \left( \prod_{i=1}^l \binom{m_i}{d_i} \binom{m_i - d_i}{a_i} \left(\sqrt{2e}\gamma_C\right)^{a_i} \gamma_C^{d_i} h^{m_i - d_i - a_i} \abs{v_{m_i}} \right) \\
&= 2 \cdot (4e \cdot \omega_C)^{l-1} \prod_{i=1}^{l} \left( \sum_{m=0}^\infty \sum_{d=0}^\infty \sum_{a=0}^\infty \binom{m}{d} \binom{m-d}{a} \left(\sqrt{2e}\gamma_C\right)^a \gamma_C^d h^{m-d-a} \abs{v_m} \right) \\
&= 2 \cdot (4e \cdot \omega_C)^{l-1} \prod_{i=1}^{l} \left( \sum_{m=0}^\infty \sum_{d=0}^\infty \binom{m}{d} \left(h + \sqrt{2e}\gamma_C\right)^{m-d} \gamma_C^d \abs{v_m} \right) \\
&= 2 \cdot (4e \cdot \omega_C)^{l-1} \prod_{i=1}^{l} \left( \sum_{m=0}^\infty \left(h + \left(\sqrt{2e} + 1\right) \gamma_C\right)^m \abs{v_m} \right) \\
&= 2 \cdot (4e \cdot \omega_C)^{l-1} \cdot \norm{\V}_{h'}^l
\end{split}
\end{equation}
with $h'=h + \left(\sqrt{2e} + 1\right) \gamma_C$. Finally, in the case $4e \cdot \omega_C \cdot \norm{\V}_{h'} < 1$, \eqref{eq:bound-on-wilsonian-effective-potential} and \eqref{eq:bound-on-wilsonian-effective-potential-remainder} just follow from \eqref{eq:bound-on-wilsonian-effective-potential-order-l} by using the sum of convergent geometric series.
\end{proof}

\section{Discussion and Outlook}
\label{sec:discussion}

The main difference to the proof of convergence of the tree expansion in fermionic quantum field theories, as for instance in \cite{SW}), is that two types of trees appear here. The iterative solution of the integrated Polchinski equation produces a sum over binary trees, the number of which does not grow factorially as a function of the iteration depth. However, the action of the Laplacians coming from the tree and from the exponential that creates loops does create another kind of tree which we call here the leaf tree. In analogy to the nomenclature of \cite{Rivbook},  the binary trees would be called the vertical trees of our expansion (which record how the iterative expansion proceeds) and the leaf tree would be called the horizontal tree (which forms the backbone of the connected Feynman graphs that are resummed into the action of the exponential of the Laplacian). Controlling the combinatorics of the latter is crucial for proving convergence. An essential combinatorial lemma that we have proven in the paper can be seen as a variant of Cayley's formula, as follows. Cayley's formula gives the number of labeled trees on $p$ vertices with incidence numbers $d_1, \ldots, d_p$ as
\beq
\frac{(p-2)!}{(d_1-1)! \ldots (d_p-1)!} \; .
\eeq
In tree expansions based on \cite{SW,clustering} or on the standard Brydges-Kennedy formula (see also \cite{anabosedec}), the factorials in the denominator are used to cancel the ones arising from the differentiation of the vertex functions. The factorial in the numerator is then cancelled by an overall $1/p!$ in front of the sum over all trees with $p$ vertices (see, e.g.\ \cite{SW}). 

Leaf trees are of course just special cases of general labelled trees, but our formula relates the number of leaf trees to the way they arise from the binary tree: we recall from Lemma \ref{lem:num-of-induced-leaf-trees} that
\beq
            \card{\mathfrak{L}(T, d)}
            = \frac{\prod_{v \in \vfork(T) \setminus \set{*}} \left( \sum_{w \in \mathcal{L}_T(v)} (d_w - 2) + 2 \right)}{\prod_{w \in \vleaf(T)} (d_w - 1)!}.
\neeq 
and from Lemma \ref{lem:num-of-leaf-trees-bound} the bound
\beq\label{eq:bound}
        \card{\mathfrak{L}(T, d)}
        \leq \frac{\prod_{v \in V(T) \setminus \set{*}} \card{\mathcal{L}_T(v)}}{\prod_{w \in \vleaf(T)} (d_w - 1)!} 
\eeq
When the tree in Cayley's formula is identified with our leaf tree, the denominators in these formulas are identical. It is the product in the numerator of \eqref{eq:bound} that encodes information about the binary tree that induces the leaf tree. While the cancellation of the factorials $d_L (w) !$ in \eqref{eq:derivcomb} by the ones in the denominators is similar to the one in other tree expansions, the cancellation of the product in the numerator of \eqref{eq:derivcomb} against factors arising in the integral over the $s$-variables fits specifically to the way our expansion is built up. 

Theorem \ref{thm:formal-solution} is valid for a general (differentiable) dependence of the covariance $C$ on the parameter $s$, while Theorem \ref{thm:convergence-of-formal-solution} includes the hypothesis that $C(s)$ depends linearly on $s$. When using RG as a multiscale method, this is still useful for controlling a single step in an iteration for discrete times. Specifically, suppose that $C (\sigma) = \int_0^\sigma \dot C(\tau) \; \dd \tau$ (where $C(\tau)$ does not depend linearly on $\tau$) and the Wilsonian effective action $\G(\sigma) = \G (\Psi, C(\sigma), \V)$ is defined by the Polchinski equation. If the limit of interest is $ \sigma \to \infty$, one can try to take this limit as the limit of the discrete sequence $\G_j (\Psi) = \G(\Psi, C(\sigma_j), \V)$, where the flow times $\sigma_j$ form an increasing sequence that tends to infinity for $j \to \infty$. The sequence starts at $\G_0 = \V$ and for $j \ge 1$, $\G_j$ is obtained from $\G_{j-1}$ by
\beq
\G_j = \G(\sigma_j) = \G \left(\Psi, C_j, \G_{j-1} \right) \; .
\eeq
where $C_j = C(\sigma_j) - C(\sigma_{j-1})$. 
By Theorem \ref{thm:convergence-of-formal-solution}, the linear interpolation $C_j(s) = s C_j $ with $s \in [0,1]$ implies that $\G_j$ is analytic in $\G_{j-1}$ in a norm that depends on $j$ via the determinant and decay constants of $C_j$. If the $\sigma_j$ are chosen properly, this can be iterated to analyze the convergence of $\G_j$ as $j \to \infty$. 

Still, it would be nice to have a method that does not involve the discrete sequence of $\sigma_j$. It is possible to generalize Theorem \ref{thm:convergence-of-formal-solution} to a general $s$-dependence of $C(s)$ \cite{inprogress}. Moreover, it is in general necessary to include renormalization effects, in particular, separate out the flow of RG-relevant terms in $\G$. This can be done for two-point terms (which are usually the most relevant ones in that they grow fastest in the flow) by a modification of Polchinski's equation in which terms that are quadratic in $\Psi$ in $\G$ are continuously absorbed in the covariance during the flow \cite{dynRG}. This modifies the Polchinski equation by additional terms that are linear in $\G$, so that the nonlinearity remains quadratic, hence a more general binary tree expansion applies. Bounds on the thus modified function $\G$ are under investigation \cite{inprogress}.

\appendix

\section{Some notions and notations from graph theory}
\label{app:graph-theory}

This appendix closely follows \cite{graphTheory}.

\begin{definition}[Graph]
A \textit{(finite simple undirected) graph} is a pair $G=(V,E)$, where $V$ is a finite set and
\begin{equation}
E \subseteq \set{\set{x,y} \subseteq V : x \neq y}.
\end{equation}
$V=V(G)$ is called the vertex set of $G$ and its elements are called \textit{vertices} of $G$. Denote by $v(G)= \card{V(G)}$ the number of vertices of $G$. $E=E(G)$ is called the edge set of $G$ and its elements are called \textit{edges} of $G$. Denote by $e(G) = \card{E(G)}$ the number of edges of $G$.

For a finite set $V$, use the notation
\begin{equation}
\Gr(V) = \set{G:\text{$G$ is a graph with $V(G)=V$}}.
\end{equation}
\end{definition}

We introduce two special kinds of graphs.

\begin{definition}[Path]
A graph $P$ is called a \textit{path of length $n$} iff $v(P)=n$ and the vertices of $P$ can be enumerated like $v_1, \dots, v_n \in V(T)$, such that $V(T) = \set{v_1, \dots, v_n}$ and $E(T) = \set{\set{v_1,v_2}, \set{v_2,v_3}, \dots, \set{v_{n-1}, v_n}}$.
\end{definition}

\begin{definition}[Cycle]
A graph $C$ is called a \textit{cycle of length $n$} iff $v(C)=n$ and the vertices of $C$ can be enumerated like $v_1, \dots, v_n \in V(C)$, such that $V(C) = \set{v_1, \dots, v_n}$ and $E(T) = \set{\set{v_1,v_2}, \set{v_2,v_3}, \dots, \set{v_{n-1}, v_n}, \set{v_n, v_1}}$.
\end{definition}

\begin{definition}[Subgraphs and modifications of graphs]
Let $G$ be a graph.
\begin{enumerate}[label=(\roman*)]
\item A graph $H=(V',E')$ is called \textit{subgraph} of $G$ iff $V' \subseteq V(G)$ and $E' \subseteq E(G)$. Write $H \subseteq G$.
\item For $V' \subseteq V(G)$, the \textit{subgraph induced by $V'$} is the graph $G[V']=(V',E')$ with
\begin{equation}
E' = \set{e\in E : e \subseteq V'} \;.
\end{equation}
A graph $H$ is called an \textit{induced subgraph} of $G$ iff some $V' \subseteq V(G)$ exists such that $H = G[V']$.
\item Let $V' \subseteq V(G)$, $v \in V(G)$, $E' \subseteq E(G)$, and $e \in E(G)$. Define:
\begin{align}
G - V' &= G[V(G) \setminus V'], \\
G - v  &= G - \set{v}, \\
G - E' &= (V(G), E(G) \setminus E'), \\
G - e  &= G - \set{e}.
\end{align}
\item For two subgraphs $H,H' \subseteq G$, define:
\begin{align}
G = H \cup H'
\quad:\Leftrightarrow\quad
&V(G) = V(H) \cup V(H') \notag \\
&\text{and } E(G) = E(H) \cup E(H'), \\
G = H \discup H'
\quad:\Leftrightarrow\quad
&V(G) = V(H) \discup V(H') \notag \\
&\text{and } E(G) = E(H) \discup E(H').
\end{align}
\end{enumerate}
\end{definition}

\begin{definition}[Properties of graphs]
\label{def:properties-of-graphs}
Let $G$ be a graph.
\begin{enumerate}[label=(\roman*)]
\item The \textit{neighbourhood} of $v \in V(G)$ in $G$ is defined as
\begin{equation}
N_G(v)=\set{w\in V(G) : \set{v,w} \in E(G)}.
\end{equation}
The \textit{neighbourhood} of $V' \subseteq V(G)$ in $G$ is defined as
\begin{equation}
N_G(V') = \bigcup_{v \in V'} N_G(v).
\end{equation}
\item The \textit{degree} of $v \in V(G)$ in $G$ is defined as $d_G(v) = \# N_G(v)$.
\item $G$ is called \textit{connected} iff for all $v,w \in V(G)$, the graph contains a path containing $v$ and $w$. For a finite set $V$, use the notation
\begin{equation}
\Gr_c(V) = \set{G \in \Gr(V) : \text{$G$ is connected}}.
\end{equation}
\item The \textit{(connected) component} of $v \in V(G)$ in $G$ is defined as
\begin{equation}
C_v = \set{w \in V(G) : \text{$\exists$ a path in $G$ connecting $v$ and $w$}}.
\end{equation}
A set $D \subseteq E(G)$ is called a \textit{(connected) component} of $G$ iff there exists $v \in V(G)$ such that $D=C_v$.
\item $G$ is called \textit{acyclic} iff it does not contain any cycles.
\end{enumerate}
\end{definition}

Sometimes, it is useful to be able to enumerate the neighbours of every vertex; therefore, we introduce the following notation.

\begin{notation}[Enumeration of neighbours]
\label{not:enumeration-of-neighbours}
Let $G$ be a graph such that there is a total order on $V(G)$. For every $v \in V(G)$, let
\begin{equation}
\iota_{G,v}:\ N_G(v) \to [d_G(v)]
\end{equation}
be the unique bijective mapping such that for all $u,w \in N_G(v)$ with $u < w$, it holds true that $\iota_{G,v}(u) < \iota_{G,v}(w)$.
\end{notation}

\begin{lemma}[Handshaking Lemma]
\label{lem:handshaking-lemma}
A graph $G$ always satisfies
\begin{equation}
\sum_{v \in V(G)} d_G(v) = 2 e(G).
\end{equation}
\end{lemma}

\begin{proof}
Induction on $v(G)$.
\end{proof}

\begin{definition}[Tree] A graph $T$ is called a \textit{tree} iff $T$ is connected and acyclic. For a finite set $V$ and $n \in \N$, use the notation:
\begin{equation}
\mathfrak{T} = \set{T : \text{$T$ is a tree}},\ 
\mathfrak{T}(V) = \set{T \in \mathfrak{T} : V(T) = V}
\text{ and }
\mathfrak{T}_n = \mathfrak{T}([n]).
\end{equation}
In a tree $T$, there is exactly one path between any two vertices $u,v \in V(T)$ (if there was more than one such path, $T$ would contain a cycle); call this path $P_{u,v}(T) \subseteq T$.
\end{definition}

\begin{remark}
A tree $T$ always satisfies $e(T) = v(T) - 1$.
\end{remark}

\section{Proof of combinatorial statements from Section~\ref{sec:combinatorics-of-graphs}}
\label{app:combinatorial-proofs}

We begin by proving Lemma~\ref{lem:num-of-induced-leaf-trees}.

\begin{proof}[Proof of Lemma~\ref{lem:num-of-induced-leaf-trees}]

(i) First, we derive a neat property of induced leaf trees on binary trees. Let $l > 1$ and $T \in \mathfrak{B}_l$ be a binary tree. If we remove the root $*$ from $T$, we get two smaller binary trees $T_1$ and $T_2$ rooted at $c_T^{(1)}(*)$ and $c_T^{(2)}(*)$ respectively (see Figure~\ref{fig:induction-root}). For a given leaf tree $L \in \mathfrak{L}(T)$ and $i \in \set{1,2}$, set
\begin{equation}
L_i = L\left[\mathcal{L}_T\left(c_T^{(i)}(*)\right)\right] \; .
\end{equation}
Here $L_1$ and $L_2$ are induced leaf trees of $T_1$ and $T_2$ respectively. The original leaf tree now satisfies
\begin{equation}
E(T) = E(T_1)\ \dot{\cup}\ E(T_2)\ \dot{\cup}\ \set{e},
\end{equation}
where $e \in E(T)$ is an edge between a vertex in $\mathcal{L}_T\left(c_T^{(1)}(*)\right)$ and a vertex in $\mathcal{L}_T\left(c_T^{(2)}(*)\right)$. But this means that the following recursion holds:
\begin{equation}
\label{eq:number-induced-leaf-trees-recursion}
\begin{split}
\card{\mathfrak{L}(T)} = &\card{\mathfrak{L}(T_1)} \cdot \card{\mathfrak{L}(T_2)} \\
&\quad \cdot \left( \card{\mathcal{L}_T\left(c_T^{(1)}(*)\right)} \cdot \card{\mathcal{L}_T\left(c_T^{(2)}(*)\right)} \right),
\end{split}
\end{equation}
where the last factor denotes the number of possibilities to choose $e$. We can now prove \eqref{eq:num-of-induced-leaf-trees-1} by induction on the number of leaves, starting at the the binary tree with one leaf. The induction hypothesis gives us for $i \in \set{1,2}$:
\begin{equation}
\card{\mathfrak{L}(T_i)} = \prod_{v \in V(T_1) \setminus \set{c_T^{(i)}(*)}} \card{\mathcal{L}_T(v)}.
\end{equation}
Together with the recursion formula \eqref{eq:number-induced-leaf-trees-recursion}, this completes the induction step.

\medskip
 (ii) Induction on $l$, with inductive hypothesis \eqref{eq:num-of-induced-leaf-trees-2}.
 Let $l \in \N$, $T \in \mathfrak{B}_l$, and $d = (d_w)_{w \in \vleaf(T)}$ with $d_w \geq 1$ for all $w \in \vleaf(T)$. If $l(T) \in \set{1,2}$, the statement is clear; hence, let $l(T) \geq 3$ and assume the inductive hypothesis to hold for all trees $\theta$ with $l(\theta) < l (T)$. Let $\tilde T = T - \vleaf(T)$. Because $l \ge 3$, $\tilde T$ is nonempty, and it is again a tree. Choose a leaf $v \in V(\tilde T) \setminus \set{*}$ of $\tilde T$ and for $i \in \set{1,2}$ let $w_i = c_T^{(i)}(v)$. Now, $w_1$ and $w_2$ are leaves of $T$ whose common parent vertex is $v$. Let $T' = T - w_1 - w_2$; this is a binary tree with
\begin{equation}\label{eq:forksTprime}
\vfork(T') = \vfork(T) \setminus \set{v}
\end{equation}
and
\begin{equation}
\vleaf(T') =  \set{v} \cup \vleaf(T) \setminus \set{w_1,w_2}
\end{equation}
(see Figure~\ref{fig:induction-leaf}). Define the sequence
\begin{equation}\label{eq:dprimedef}
\begin{split}
&d' = (d'_w)_{w \in \vleaf(T')} \\
&\text{with } d'_w = \begin{cases}
d_w & \text{if $w \neq v$} \\
d_{w_1} + d_{w_2} - 2 & \text{if $w = v$}\;.
\end{cases}
\end{split}
\end{equation}
Further define a mapping
\begin{equation}
f:\ \vleaf(T) \to \vleaf(T'),\ w \mapsto \begin{cases}
w & \text{if $w \notin \set{w_1, w_2}$} \\
v & \text{if $w \in \set{w_1, w_2}$}\;.
\end{cases}
\end{equation}
If $\mathfrak{L}(T,d)= \emptyset$, the statement is trivial. Otherwise, consider a leaf tree $L \in \mathfrak{L}(T,d)$. Define $L' = (\vleaf(T'), E')$ with
\begin{equation}
E' = \set{\set{f(a), f(b)} : \set{a,b} \in E(L)} \setminus \set{v},
\end{equation}
where we subtract $\set{v}$ because the edge $\set{w_1,w_2} \in E(L)$ gets mapped to $\set{v}$. By construction, we have $L' \in \mathfrak{L}(T',d')$ with $l(T') = l(T) - 1$. Hence, what we have effectively constructed is a mapping
\begin{equation}
F:\ \mathfrak{L}(T,d) \to \mathfrak{L}(T',d')\; , \quad L \mapsto F(L) = L' \;.
\end{equation}
We now claim that for every $L' \in \mathfrak{L}(T',d')$,
\begin{equation}
\card{F^{-1}(L')}
= \binom{d'_v}{d_{w_1} - 1}
= \frac{d'_v!}{(d_{w_1} - 1)! \cdot (d_{w_2} - 1)!}.
\end{equation}
This holds because there are $d'_v$ edges connecting to $v$ in $L'$, of which $d_{w_1} - 1$ have to be chosen to connect to $w_1$ in $L$ and $d_{w_2} - 1$ have to be chosen to connect to $w_2$ in $L$. Thus
\beq
\card{\mathfrak{L}(T,d)}
=
\card{\mathfrak{L}(T',d')} \;\frac{d'_v!}{(d_{w_1} - 1)! \; (d_{w_2} - 1)!} \; .
\eeq
The inductive hypothesis applies to $T'$, so
\beq
\card{\mathfrak{L}(T',d')}
=
\frac{\prod_{u \in \vfork(T') \setminus \set{*}} \left( \sum_{w \in \mathcal{L}_{T'}(u)} ({ d'_w} - 2) + 2 \right)}{\prod_{w \in \vleaf(T')} ({ d'_w} - 1)!}
\eeq
We now rewrite numerator and denominator in terms of $T$. For each factor in the product in the numerator, the equation $d'_v - 2 = (d_{w_1} + d_{w_2} - 2) - 2 = (d_{w_1} - 2) + (d_{w_2} - 2)$ implies for all $u \in \vfork(T') \setminus \set{*}$:
\begin{equation}
\sum_{w \in \mathcal{L}_{T'}(u)} ({ d'_w} - 2) + 2 = \sum_{w \in \mathcal{L}_{T}(u)} (d_w - 2) + 2 \;.
\end{equation}
Moreover, by \eqref{eq:forksTprime} and by \eqref{eq:dprimedef}, we can extend the product in the numerator to all $u \in \vfork(T) \setminus \set{*}$ if we divide by a factor
\beq
\sum_{w \in \mathfrak{L}_T(v)} (d_w - 2) + 2 = d_{w_1} + d_{w_2} - 2 = d'_v \;.
\eeq
The induction step is now completed by noting that
\beq
\frac{(d'_v-1)!}{(d_{w_1} - 1)! \; (d_{w_2} - 1)!} \; \frac{1}{\prod_{w \in \vleaf(T')} ({ d'_w} - 1)!} 
=
\frac{1}{\prod_{w \in \vleaf(T)} (d_w - 1)!} \; .
\eeq
\end{proof}

In order to prove Lemma~\ref{lem:num-of-leaf-trees-bound}, we first need another definition.

\begin{definition}[Leaf tree with extra vertices]
Let $T \in \mathfrak{B}$ be a binary tree and $k \in \N_0$. We say that a graph $L=(V,E)$ is an \textit{(induced) leaf tree of $T$ with $k$ extra vertices} iff there exists a leaf tree $L' \in \mathfrak{L}(T)$ and a mapping $f:\ [k] \to \vleaf(T)$ such that\footnote{$\sqcup$ denotes the disjoint union} 
\begin{equation}
V = \vleaf(T) \sqcup [k]
\text{ and }
E = E(L') \cup \set{\set{i,f(i)} : i \in [k]}.
\end{equation}
$L$ is indeed a tree, and every $i \in [k]$ is a leaf of $L$. 
We write
\begin{equation}
\mathfrak{L}_k(T) = \set{L : \text{$L$ is an induced leaf tree on $T$ with $k$ extra vertices}},
\end{equation}
and:
\begin{equation}
\mathfrak{L} = \set{(l,T,k,L) : l \in \N,\ T \in \mathfrak{B}_l,\ k \in \N_0,\ L \in \mathfrak{L}_k(T)}.
\end{equation}
\end{definition}

Intuitively, two extra vertices arise when, in the operation of removing the root of the binary tree, one of the lines of the induced leaf tree gets cut into two ``half-lines''. In the nomenclature of Feynman graphs in QFT, one would say that the graphs with $k$ extra vertices are graphs with $k$ external legs, i.e.\ every external leg ends at one of the extra vertices.

Note that $\mathfrak{L}_0(T) = \mathfrak{L}(T)$ for all $T \in \mathfrak{B}$. Next, we prove two auxiliary lemmas.

\begin{lemma}
\label{lem:degree-sums-in-leaf-trees-with-extra-vertices}
For $T \in \mathfrak{B}$, $k \in \N_0$, and $L \in \mathfrak{L}_k(T)$, we have:
\begin{equation}
\label{eq:degree-sums-in-leaf-trees-with-extra-vertices}
\sum_{w \in \mathcal{L}_T(*)}(d_L(w) - 2) + 2 = k.
\end{equation}
\end{lemma}

\begin{proof}
By Lemma~\ref{lem:handshaking-lemma}, 
\begin{equation}
 \sum_{v \in V(L)} d_L(v) 
 = 2 e(L) = 2 (v(L) - 1) = 2 (l(T) + k - 1).
\end{equation}
and by definition, $\sum_{v \in V(L)} d_L(v) = \sum_{w \in \leavesof{T}{*}} d_L(w) + k$.
Comparing and rearranging this equation and using $l(T) = \card{\leavesof{T}{*}}$ yields Equation~\eqref{eq:degree-sums-in-leaf-trees-with-extra-vertices}.
\end{proof}

\begin{lemma}
\label{lem:bound-for-leaf-trees-with-extra-vertices}
Consider the mapping:
\begin{equation}
F:\ \mathfrak{L} \to \R_{\geq 0},\ 
(l,T,k,L) \mapsto \frac{\prod_{v \in \vfork(T)} \left(\sum_{w \in \mathcal{L}_T(v)} (d_L(w) - 2) + 2\right)}{\prod_{v \in V(T) \setminus \set{*}} \card{\mathcal{L}_T(v)}}.
\end{equation}
Then for all $(l,T,k,L) \in \mathfrak{L}$ with $k \geq 1$
\begin{equation}
F(l,T,k,L) \leq \binom{k+l-2}{k-1}\;.
\end{equation}
\end{lemma}

\begin{proof}
Let $(l,T,k,L) \in \mathfrak{L}$ with $k \geq 1$. For $l = 1$, the statement is clear; hence, let $l \geq 2$.
If we remove the root $*$ from $T$, we get two smaller binary trees $T_1$ and $T_2$ rooted at $v_1 = c_T^{(1)}(*)$ and $v_2 = c_T^{(2)}(*)$ respectively (see Figure~\ref{fig:induction-root}). Let $l_1=l(T_1)$ and $l_2=l(T_2)$. Then $l_1 \ge1$ and $l_2 \ge 1$. For $i \in \set{1,2}$, define:
\begin{equation}
L_i = L\left[ \vleaf(T_i) \cup N_L(\vleaf(T_i)) \right]
\text{ and }
k_i = \card{N_L(\vleaf(T_i))} \geq 1.
\end{equation}
Note that $L_1$ and $L_2$ are leaf trees of $T_1$ and $T_2$ with $k_1$ and $k_2$ extra vertices respectively.\footnote{To match our exact definition of leaf trees with extra vertices, one would need to rename some vertices. However, this does not matter here.} Also note that $k_1 + k_2 = k + 2$, and $d_L(w) = d_{L_i}(w)$ for every $w \in \vleaf(T_i)$ and $i \in \set{1,2}$. Using the inductive hypothesis and Lemma~\ref{lem:degree-sums-in-leaf-trees-with-extra-vertices}, we get
\begin{equation}
\label{eq:bound-for-leaf-trees-with-extra-vertices-induction-step-1}
\begin{split}
F(l,T,k,L)
&= \frac{\prod_{v \in \vfork(T)} \left(\sum_{w \in \mathcal{L}_T(v)} (d_L(w) - 2) + 2\right)}{\prod_{v \in V(T) \setminus \set{*}} \card{\mathcal{L}_T(v)}} \\
&= \left( \prod_{i \in \set{1,2}} \frac{\prod_{v \in \vfork(T_i)} \left(\sum_{w \in \mathcal{L}_{T_i}(v)} (d_{L_i}(w) - 2) + 2\right)}{\prod_{v \in V(T_i) \setminus \set{*}} \card{\mathcal{L}_{T_i}(v)}} \right) \cdot \frac{k}{l_1 \cdot l_2} \\
&= F(l_1,T_1,k_1,L_1) \cdot F(l_2,T_2,k_2,L_2) \cdot \frac{k}{l_1 \cdot l_2} \\
&\leq \binom{k_1 + l_1 - 2}{k_1 - 1} \cdot \binom{k_2 + l_2 - 2}{k_2 - 1} \cdot \frac{k}{l_1 \cdot l_2}.
\end{split}
\end{equation}
Because $l_i \ge 1$,  $l_1 \cdot l_2 \geq l_1 + l_2 - 1$, so
\begin{equation}
\label{eq:bound-for-leaf-trees-with-extra-vertices-induction-step-2}
F(l,T,k,L) \leq \frac{k}{l_1 + l_2 - 1} \cdot \binom{k_1 + l_1 - 2}{k_1 - 1} \cdot \binom{k_2 + l_2 - 2}{k_2 - 1}.
\end{equation}
We can bound this expression by a sum of positive terms and then use Vandermonde's identity on this sum,
\begin{equation}
\label{eq:bound-for-leaf-trees-with-extra-vertices-induction-step-3}
\begin{split}
F(l,T,k,L)
&\leq \frac{k}{l_1 + l_2 - 1} \cdot \sum_{j \ge 0 } \binom{k_1 + l_1 - 2}{(k_1 - 1) + (k_2 - 1) - j} \cdot \binom{k_2 + l_2 - 2}{j} \\
&= \frac{k}{l - 1} \cdot \binom{(k_1 + l_1 - 2) + (k_2 + l_2 - 2)}{(k_1 - 1) + (k_2 - 1)} \\
&= \frac{k}{l - 1} \binom{(k_1 + k_2 - 2) + (l_1 + l_2) - 2}{k_1 + k_2 - 2} \\
&= \frac{k}{l - 1} \binom{k + l - 2}{k}
= \binom{k + l - 2}{k - 1}.
\end{split}
\end{equation}
\end{proof}

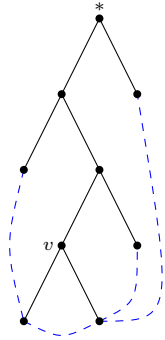
\begin{figure}[h!]
\tiny
\begin{tikzpicture}
\node[anchor=south] at (1,4) {$*$};
\draw (1,4) -- (0.5,3);
\draw (1,4) -- (1.5,3);

\draw (0.5,3) -- (0,2);
\draw (0.5,3) -- (1,2);

\draw (1,2) -- (0.5,1) node[anchor=east] {$v$};
\draw (1,2) -- (1.5,1);

\draw (0.5,1) -- (0,0);
\draw (0.5,1) -- (1,0);

\draw[blue,dashed] (0,2) .. controls (-0.25,1) .. (0,0);
\draw[blue,dashed] (0,0) .. controls (0.5,-0.25) .. (1,0);
\draw[blue,dashed] (1,0) .. controls (1.5,0.25) .. (1.5,1);
\draw[blue,dashed] (1,0) .. controls (2,0) .. (1.5,3);

\node[circle, fill=black, inner sep=0, minimum size=3pt] at (1,4) {};
\node[circle, fill=black, inner sep=0, minimum size=3pt] at (0.5,3) {};
\node[circle, fill=black, inner sep=0, minimum size=3pt] at (1.5,3) {};
\node[circle, fill=black, inner sep=0, minimum size=3pt] at (0,2) {};
\node[circle, fill=black, inner sep=0, minimum size=3pt] at (1,2) {};
\node[circle, fill=black, inner sep=0, minimum size=3pt] at (0.5,1) {};
\node[circle, fill=black, inner sep=0, minimum size=3pt] at (1.5,1) {};
\node[circle, fill=black, inner sep=0, minimum size=3pt] at (0,0) {};
\node[circle, fill=black, inner sep=0, minimum size=3pt] at (1,0) {};
\end{tikzpicture}
\normalsize
\caption{Example of an extremal leaf tree.}
\label{fig:extremal-binary-tree}
\end{figure}

\begin{remark}
\label{rem:def-of-extremal-binary-trees}
We say that a binary tree $T \in \B$ is \textit{extremal} iff there exists a vertex $v \in \vfork(T)$ such that
\begin{equation}
\vfork(T) = V(P_{*,v}(T)),
\end{equation}
where $P_{*,v}(T)$ is the unique path in $T$ that connects $*$ and $v$. We call $v$ the \textit{extremal vertex} of $T$.
If $T \in \B$ is an extremal binary tree with extremal vertex $v \in \vfork(T)$, we say that $L \in \L_k(T)$ is an \textit{extremal leaf tree} iff $d_L(u) = 1$ for all $u \in \vleaf(T) \setminus \leavesof{T}{v}$ (see Figure~\ref{fig:extremal-binary-tree}).

Now, we can make the following statement: The bound for $F(l,T,k,L)$ from Lemma~\ref{lem:bound-for-leaf-trees-with-extra-vertices} is tight iff $T$ is an extremal binary tree and $L$ is an extremal leaf tree. (This can be easily seen in the induction step.) Otherwise, the inequalities in Equations~\eqref{eq:bound-for-leaf-trees-with-extra-vertices-induction-step-2} and \eqref{eq:bound-for-leaf-trees-with-extra-vertices-induction-step-3} lead to an overestimation.
\end{remark}

The preceding auxiliary lemma now finally allows us to prove Lemma~\ref{lem:num-of-leaf-trees-bound}:

\begin{proof}[Proof of Lemma~\ref{lem:num-of-leaf-trees-bound}]
The denominator in \eqref{eq:num-of-induced-leaf-trees-2} and in \eqref{eq:num-of-leaf-trees-bound} is the same, so once the indicator function is taken care of, the point is to bound the numerator of \eqref{eq:num-of-induced-leaf-trees-2}. 

Let $l \in \N$, $T \in \mathfrak{B}_l$, and $d = (d_w)_{w \in \vleaf(T)}$ with $d_w \geq 1$ for all $w \in \vleaf(T)$. If $\mathfrak{L}(T, d) = \emptyset$, the statement is trivial. Otherwise, let $L \in \mathfrak{L}(T, d)$. By Lemma~\ref{lem:handshaking-lemma}, we have $\sum_{w \in \vleaf(T)} d_w = 2e(L) = 2(l-1)$, i.e.\ $\bbbone(\sum_{w \in \vleaf(T)} d_w = 2(l-1)) = 1$.
Now, choose $w_0 \in \vleaf(T)$ and let $L' \in \mathfrak{L}_1(T)$ with
\begin{equation}
V(T') = V(T) \sqcup \set{1}
\text{ and }
E(T') = E(T) \cup \set{\set{1,w_0}}.
\end{equation}
 
By Lemma~\ref{lem:num-of-induced-leaf-trees} (ii) 
\begin{equation}
\begin{split}
\prod_{w \in \vleaf(T)} (d_w - 1)! \; \card{\mathfrak{L}(T, d)}
&= \prod_{v \in \vfork(T) \setminus \set{*}} \left( \sum_{w \in \mathcal{L}_T(v)} (d_L(w) - 2) + 2 \right)
\\
&\leq \prod_{v \in \vfork(T)} \left( \sum_{w \in \mathcal{L}_T(v)} (d_{L'}(w) - 2) + 2 \right)
\\
&= \prod_{v \in V(T) \setminus \set{*}} \card{\mathcal{L}_T(v)} \cdot F(l,T,1,L') \; .
\end{split}
\end{equation}
By Lemma~\ref{lem:bound-for-leaf-trees-with-extra-vertices}, $F(l,T,1,L') \le 1$, which completes the proof.  
\end{proof}

\section{Selected proofs for Section~\ref{subsec:determinants} and \ref{subsec:trimming-the-tree}}
\label{app:selected-proofs}

\begin{proof}[Proof of Lemma~\ref{lem:laplacians-and-determinants}]
Compute:
\begin{equation}
\begin{split}
e^{\Delta_\Mcov} \prod_{X \in D} \psi(X)
&= \sum_{k=0}^\infty \frac{1}{k!} \left(\fdv{}{\Psi},\Mcov\fdv{}{\Psi}\right)^k \prod_{X \in D} \psi(X) \\
&= \sum_{k=0}^\infty \sum_{\substack{A,A' \subseteq D \\ \card{A}=\card{A'}=k \\ A \cap A' = \emptyset}} \tilde{\varepsilon}_D^{A,A'} \left( \prod_{X \in D \setminus (A \cup A')} \psi(X) \right) \\
&\qquad \cdot \frac{1}{k!} \left(\fdv{}{\Psi^{(1)}},\Mcov\fdv{}{\Psi^{(2)}}\right)^k \prod_{X \in A} \psi^{(1)}(X) \prod_{X \in A'} \psi^{(2)}(X),
\end{split}
\end{equation}
where we have (using the notation $A=\set{a_1, \dots, a_k}$ and $A'=\set{a'_1, \dots, a'_k}$):\footnote{Note that all products in these equations are actually ordered products and that the ordering of $a_1, \dots, a_k$ and $a'_1, \dots, a'_k$ has to respect this order. This does not matter here though, because we are not interested in overall signs; these are absorbed into the $\varepsilon_D^{A,A'}$.}
\begin{equation}
\begin{split}
&\frac{1}{k!} \left(\fdv{}{\Psi^{(1)}},\Mcov\fdv{}{\Psi^{(2)}}\right)^k \prod_{X \in A} \psi^{(1)}(X) \prod_{X \in A'} \psi^{(2)}(X) \\
= &-\frac{1}{k!} \sum_{\substack{X_1, \dots, X_k \in A \\ X'_1, \dots, X'_k \in A'}} \left(\prod_{i=1}^k \Mcov_{X_iX'_i}\right) \\
&\qquad \cdot \left(\fdv{}{\psi(X_1)} \dots \fdv{}{\psi(X_k)} \prod_{X\in A} \psi(X) \right)
\left(\fdv{}{\psi(X'_1)} \dots \fdv{}{\psi(X'_k)} \prod_{X\in A'} \psi(X) \right) \\
= & -\frac{1}{k!} \sum_{\sigma,\tau \in S_k} \sgn(\sigma)\; \sgn(\tau)\; \Mcov_{a_{\sigma(1)}a'_{\tau(1)}} \dots \Mcov_{a_{\sigma(k)}a'_{\tau(k)}} \\
= & - \sum_{\sigma \in S_k} \sgn(\sigma)\; \Mcov_{a_1a'_{\sigma(1)}} \dots \Mcov_{a_ka'_{\sigma(k)}}
= - \det(\Mcov_{A,A'}).
\end{split}
\end{equation}
This proves the lemma.
\end{proof}

\begin{proof}[Proof of Lemma~\ref{lem:trimming-the-tree}]
Define
\begin{equation}
\begin{split}
&F\Big(T,(m_q)_{q \in V(T)}\Big) \\
&= \max_{(q',i') \in D} \sup_{\tilde X \in \X} \int \Bigg( \prod_{(q,i) \in D \setminus \set{(q',i')}} \dd{X_{(q,i)}} \Bigg) \Bigg( \prod_{q=1}^p \abs{v_{m_q}\Big(X_{(q,1)}, \dots, X_{(q,m_q)}\Big)} \Bigg) \\
&\mkern135mu \cdot \Bigg( \prod_{\set{q,q'} \in E(T)} \abs{C\Big(X_{(q,i_{T,q}(q'))}, X_{(q',i_{T,q'}(q))}\Big)} \Bigg) \Bigg|_{X_{(q',i')} = \tilde X}.
\end{split}
\end{equation}
We now prove the lemma by induction on $p$. For $p=1$, the statement is true trivially. Now let $p \geq 2$. Choose $(q',i') \in D$ to maximise the expression \eqref{eq:trimming-the-tree}. Let $v \in [p]$ be a leaf of $T$ that is not $q'$, and let $w$ be the unique neighbour of $v$ in $T$, i.e.\ $w \in N_T(v)$. Let $i \in [d_T(w)]$ such that $i_{T,w}(v) = i$. Then, by taking suprema under the integrals:
\begin{equation}
\begin{split}
F\Big(T, (m_q)_{q \in V(T)}\Big)
\leq& \Bigg( \sup_{X_{(v,1)} \in \X} \int \Bigg(\prod_{i=2}^{m_v} \dd{X_{(v,i)}}\Bigg) \abs{v_{m_v}\Big(X_{(v,1)}, \dots, X_{(v,m_v)}\Big)} \Bigg) \\
& \cdot \Bigg( \sup_{X_{(w,i)} \in \X} \int \dd{X_{(v,1)}}\; \abs{C\Big(X_{(v,1)}, X_{(w,i)}\Big)} \Bigg) \\
& \cdot F\Big(T-v, (m_q)_{q \in V(T-v)}\Big).
\end{split}
\end{equation}
The first line is bounded by $\abs{v_{m_v}}$, the second line by $\abs{C}$.
The statement now follows directly from the induction hypothesis.
\end{proof}

\noindent
{\bf Acknowledgement. } This work is supported by Deutsche Forschungsgemeinschaft (DFG, German Research Foundation) under Germany's Excellence Strategy  EXC-2181/1 - 390900948 (the Heidelberg STRUCTURES Cluster of Excellence). 

\newcommand{\reftitel}[1]{{\em #1}, }

\end{document}